\documentclass[11pt]{article}
\usepackage[margin=1in]{geometry}
\usepackage{amsmath,amssymb,amsthm,mathtools,bm}
\usepackage{bbm}
\usepackage{mathrsfs}
\usepackage{enumitem}
\usepackage{booktabs}
\usepackage{hyperref}
\usepackage{natbib}
\usepackage{setspace}
\usepackage{multirow}
\usepackage{array}
\usepackage{color}
\newtheorem{theorem}{Theorem}[section]

\newtheorem{lemma}{Lemma}[section]

\newtheorem{remark}{Remark}[section]
\newtheorem{assumption}{Assumption}[section]

\newcommand{\R}{\mathbb R}
\newcommand{\de}{\delta}

\newcommand{\be}{\bm e}

\def\T{\top}
\def\E{\mathrm{E}}
\def\Pr{\mathrm{P}}
\def\var{\mathrm{Var}}

\def\cov{\mathrm{Cov}}
\def\diag{\mathrm{diag}}

\def\tr{\mathrm{tr}}

\newcommand{\op}{\mathrm{o}_p}
\newcommand{\Op}{\mathrm{O}_p}
\def\cp{\mathop{\rightarrow}\limits^{p}}
\def\cd{\mathop{\rightarrow}\limits^{d}}
\def\mR{\mathbb{R}}
\def\mZ{\mathbb{Z}}

\def\A{{\bf A}}
\def\B{{\bf B}}

\def\D{{\bf D}}

\def\I{{\bf I}}

\def\R{{\bf R}}

\def\U{{\bf U}}

\def\X{{\bf X}}
\def\Y{{\bf Y}}
\def\Z{{\bf Z}}

\def\a{{\bm a}}
\def\b{{\bm b}}
\def\e{{\bm e}}
\def\f{{\bm f}}
\def\h{{\bm h}}
\def\r{{\bm r}}

\def\u{{\bm u}}
\def\v{{\bm v}}

\def\z{{\bm z}}

\def\bmS{{\bf \Sigma}}
\def\bG{{\bf \Gamma}}

\def\bO{{\bf \Omega}}

\def\bmb{{\bm \beta}}

\def\bmeta{{\bm\eta}}

\def\bml{{\bm \lambda}}
\def\bmd{{\bm \delta}}

\def\bzero{{\bm 0}}
\def\bone{{\bm 1}}

\newcommand{\norm}[1]{\lVert #1 \rVert}

\newcommand{\wt}{\widetilde}
\newcommand{\wh}{\widehat}
\newcommand{\wb}{\overline}
\newcommand{\wk}{\check}
\newcommand{\mc}{\mathcal}

\newcommand{\bP}{\mathbf P}
\newcommand{\bM}{\mathbf M}

\usepackage[thicklines]{cancel}
\usepackage{xcolor}

\title{Testing Alpha in High-Dimensional Conditional Time-Varying Factor Models with Dependent Observations}
\author{Long Feng, Huifang Ma and Zhaojun Wang\\
Nankai University}
\date{\today}

\begin{document}
\maketitle

\begin{abstract}
This paper studies alpha testing in a high-dimensional conditional time-varying factor model with temporally dependent observations. Both factor loadings and alpha processes are allowed to vary smoothly over time, and the cross-sectional dimension may be comparable to or larger than the sample size. Using a B-spline sieve method, we develop a sum-type test for dense alternatives, a max-type test for sparse alternatives, and a Cauchy combination test for adaptive inference. On the theoretical side, we derive explicit stochastic expansions for the estimated average alphas, establish asymptotic normality of the sum statistic, and develop the extreme-value limit theory for the max statistic by showing its Gumbel convergence under temporal dependence together with the validity of block-bootstrap calibration. We further prove asymptotic independence between the sum and max statistics and thereby justify the Cauchy combination test. Simulation results demonstrate that the proposed procedures achieve satisfactory size control and competitive power across a wide range of dense and sparse alternatives. An empirical application further illustrates the usefulness of the proposed methods in testing asset-pricing models with time-varying structure.
\end{abstract}


\section{Introduction}
\citet{sharpe1964capital} and \citet{lintner1965valuation} laid the foundation for modern asset-pricing theory by linking expected returns to systematic risk in the capital asset pricing model. Subsequent empirical work moved beyond the single-factor benchmark toward multifactor specifications; in particular, \citet{fama1993common} showed that a small set of common factors can explain an important part of the comovement and average returns of stocks and bonds. Against this background, testing whether alphas are jointly zero remains one of the central specification problems in empirical asset pricing, because nonzero pricing errors indicate that the proposed factors fail to fully account for expected returns.

In the classical linear factor pricing model with fixed factor loadings, \citet{gibbons1989test} analyzed a test of the ex ante efficiency of a given portfolio, derived a tractable small-sample distribution for the resulting statistic, studied its power, and provided geometric as well as diagnostic interpretations; in the Gaussian setting, this test became the canonical finite-sample benchmark for testing mean--variance efficiency and the joint nullity of pricing errors. \citet{mackinlay1991using} developed tests of unconditional mean--variance efficiency within a generalized method of moments framework under relatively weak distributional assumptions, thereby extending alpha testing beyond the strict normal-i.i.d.\ environment and showing empirically that conclusions can be sensitive to the testing approach adopted. \citet{zhou1993asset} reconsidered mean--variance efficiency tests under nonnormal alternatives, especially elliptical distributions, and showed that normality-based procedures may overreject efficiency while overstating quantities such as beta and $R^2$ when normality fails. Moving further in this direction, \citet{beaulieu2007multivariate} developed exact simulation-based finite-sample tests for multivariate mean--variance efficiency in CAPM settings with possibly non-Gaussian errors, covering both unconditional and conditional specifications; in the Gaussian unconditional case, their procedure nests the classical GRS framework.

These classical methods become difficult to use in modern applications with many assets. Their validity relies on stable estimation and inversion of the residual covariance matrix, but this becomes fragile when the cross-sectional dimension is of the same order as, or even larger than, the sample size. This challenge has stimulated a large literature on high-dimensional alpha testing. 
\cite{pesaran2012testing,PesaranYamagata2024} developed a feasible high-dimensional test for zero alphas that is based on Studentized individual pricing errors and remains valid when the cross-sectional dimension is much larger than the sample size. \citet{lan2018testing} proposed a random-projection-based test for high-dimensional linear asset pricing models, which avoids direct inversion of a large residual covariance matrix and remains effective even when the idiosyncratic covariance is strongly correlated and nonsparse. \citet{gungor2013testing} developed a finite-sample distribution-free procedure for testing linear factor pricing models with a large cross section.
 \citet{feng2022alpha} proposed a high-dimensional alpha test designed for sparse alternatives and showed that sum-based procedures can be sharpened substantially when only a small subset of assets carries nonzero pricing errors. \citet{fan2015power} introduced the power-enhancement principle, which has become an important device for improving sensitivity to sparse departures while preserving tractable null behavior.  \citet{yu2023power} proposed a power-enhanced test for multi-factor asset pricing models by combining a Wald-type statistic and a maximum-type statistic through Fisher's method, and showed that the resulting procedure achieves valid size and improved power across both sparse and dense alternatives. \citet{xia2024adaptive} proposed an adaptive testing procedure for high-dimensional linear factor pricing models that is particularly effective against sparse alternatives. Their method exploits high-dimensional Gaussian approximation and simulation-based calibration to improve power when only a small fraction of assets exhibits nonzero alphas.
\citet{liu2023robust,zhao2022high} extended high-dimensional alpha testing to heavy-tailed settings and substantially broadened its empirical applicability.  

A parallel line of research has recognized that factor exposures and risk premia may vary over time rather than remain constant. In an influential critique of conditional asset-pricing explanations, \citet{lewellen2006conditional} showed that although time-varying betas can generate differences between conditional and unconditional pricing relations, the implied unconditional pricing errors are generally too small to account for major empirical anomalies. In contrast, \citet{ang2007CAPM} showed that, over the long run, a conditional one-factor CAPM with time-varying betas can explain the average returns of book-to-market portfolios much better than standard unconditional regressions, and they emphasized that ignoring time variation in loadings can distort inference on alphas and betas. Moving to large panels, \citet{gagliardini2016time} developed an econometric framework for estimating time-varying risk premia in large unbalanced cross sections of equities under conditional linear asset-pricing models, allowing for both common and asset-specific instruments. The survey of \citet{gagliardini2020estimation} synthesized this literature and highlighted that time variation in factor loadings and risk premia is a first-order feature of modern large-dimensional factor modeling.

In the high-dimensional conditional factor model, \citet{ma2020testing} proposed a high-dimensional alpha test together with a constant-coefficient test, using sieve approximation based on B-spline basis functions to handle time-varying alphas and factor loadings. Building on this framework, \citet{ma2024adaptive} introduced a maximum-type test designed for sparse alternatives, established its asymptotic independence from the earlier sum-type statistic, and proposed an adaptive combination test that is more robust when the sparsity level of alphas is unknown.

Despite this progress, one important difficulty remains insufficiently resolved: financial observations are rarely independent over time. Residuals often exhibit weak dependence, serial correlation, and conditional heteroskedasticity, and these features matter directly for the behavior of high-dimensional test statistics. If temporal dependence is ignored and the observations are treated as independent, the resulting null centering, scaling, and critical values can be seriously distorted. In such cases, large rejection probabilities may simply reflect misspecified null calibration rather than genuine evidence against the model. This issue is particularly consequential in high dimensions, where even small dependence-induced distortions can accumulate and lead to markedly misleading empirical conclusions.

Recent work on high-dimensional hypothesis testing with temporally dependent observations has begun to move the literature beyond the classical i.i.d.\ framework. For mean testing, \citet{ayyala2017mean} studied one-sample and two-sample mean vector testing when the observations form an $M$-dependent stationary process, and \citet{cho2019note} later clarified and corrected parts of the corresponding asymptotic normality arguments. \citet{wang2020selfnorm} developed a self-normalized procedure for inference on the mean of high-dimensional stationary processes, thereby avoiding direct estimation of the long-run variance in a challenging high-dimensional setting. Moving from mean testing to serial-dependence diagnostics, \citet{tsay2020serial} proposed an extreme-value-based test for serial correlations in high-dimensional time series, while \citet{chang2024martingale} developed a high-dimensional test of the martingale difference hypothesis that is able to detect nonlinear serial dependence. For two-sample problems under temporal dependence, \citet{yang2024blockwise} proposed an $\ell_\infty$-type blockwise-bootstrap test for high-dimensional time series. In the context of linear factor pricing models, \citet{ma2024dependent} showed that once serial dependence is explicitly incorporated, the size and power properties of high-dimensional alpha tests can differ substantially from their i.i.d.\ counterparts, making dependence adjustment indispensable for reliable empirical inference.

The present paper develops a dependence-aware testing framework for the high-dimensional conditional time-varying factor model. Our starting point is the observation that, in conditional models with time-varying coefficients, the main inferential difficulty is not only high dimensionality but also the interaction between serial dependence and sieve estimation. We focus on the joint nullity of the average alphas and construct testing procedures that remain valid when the observations are temporally dependent. The key message is that dependence is not a secondary refinement here: it is a fundamental ingredient of correct inference. Ignoring it may turn a nominal specification test into a systematically biased decision rule.

Our methodology contains three components. First, we construct a sum-type statistic from the estimated average alphas and calibrate it by block bootstrap methods that account for serial dependence in the projected score process. This procedure is particularly effective against dense alternatives. Second, we construct a max-type statistic using coordinatewise long-run variance normalization. This statistic is tailored to sparse alternatives, admits a Gumbel-type null approximation, and can also be implemented with dependence-aware bootstrap calibration. Third, we combine the sum-type and max-type procedures through a Cauchy combination rule in the spirit of \citet{liu2020cauchy} and \citet{long2023cauchy}. The resulting test inherits strong sensitivity to sparse signals from the max component, stable performance under denser alternatives from the sum component, and robustness across unknown sparsity regimes from the combination step.

Our contributions are threefold. First, we provide a high-dimensional testing framework for conditional time-varying factor models that remains valid under serial dependence. This directly addresses a practically important source of misspecification that can invalidate procedures based on independence. Second, we show that the additional estimation error generated by the spline sieve can be controlled sharply enough to support dependence-sensitive high-dimensional asymptotics for both sum-type and max-type statistics. Third, we develop a unified inference strategy that delivers reliable size control together with complementary power properties: the sum-type test is effective against dense alternatives, the max-type test is effective against sparse alternatives, and the Cauchy combination test remains robust when the sparsity pattern is unknown. From an empirical perspective, the main implication is straightforward. In high-dimensional conditional asset-pricing models, serial dependence should be treated as part of the null structure itself; otherwise, the resulting evidence on model adequacy may be substantially distorted.

The remainder of the paper is organized as follows. Section~\ref{sec:tests} introduces the conditional time-varying factor model, describes its B-spline approximation, and develops the proposed sum-type, max-type, and Cauchy combination tests. Section~\ref{sec:main} establishes the main asymptotic theory. Section~\ref{sec:simulation} reports simulation results, while Section~\ref{sec:empirical_sp500} presents an empirical application. Section~\ref{sec:discussion} concludes. All technical proofs are deferred to Section~\ref{sec:proofs}.

\section{Test procedures}
\label{sec:tests}
Let $R_{it}$ denote the excess return of asset $i\in\{1,\dots,N\}$ at time $t\in\{1,\dots,T\}$, and let $\f_t=(f_{t1},\dots,f_{td})^{\T}\in\mR^d$ be the observed factor vector, where $d$ is fixed. We consider the conditional time-varying factor model
$$
R_{it}=\alpha_i(t/T)+\beta_i(t/T)^{\T} \f_t+e_{it}, \qquad i=1,\dots,N,\ t=1,\dots,T.
$$
Write $\bmb_i(u)=(\beta_{i1}(u),\dots,\beta_{id}(u))^{\T}$. Following \citet{ma2020testing} and \citet{ma2024adaptive}, define the average alpha
$$
\delta_i=T^{-1}\sum_{t=1}^T \alpha_i(t/T)
$$
and the centered varying component
$$
g_i(t/T)=\alpha_i(t/T)-\delta_i.
$$
Then the model can be rewritten as
$$
R_{it}=\delta_i+g_i(t/T)+\bmb_i(t/T)^{\T} \f_t+e_{it}.
$$
The hypothesis of interest is
\begin{equation}
H_0:\ \delta_1=\cdots=\delta_N=0
\qquad\text{versus}\qquad
H_1:\ \delta_i\neq 0\ \text{for some }i.
\label{eq:hyp}
\end{equation}
Thus the test concerns the time-average pricing error, not the pointwise value of $\alpha_i(u)$.

Let $\mc S_{q,L}$ be the spline space of order $q\ge 2$ with $p$ interior knots on $[0,1]$, so $L=p+q$. Let $\B(u)=(B_1(u),\dots,B_L(u))^{\T}$ denote the normalized B-spline basis, and define the centered basis
$$
\wt{\B}(u)=\B(u)-T^{-1}\sum_{t=1}^T \B(t/T).
$$
For each $i$, approximate
$$
g_i(u)\approx \bml_{i0}^{\T} \wt{\B}(u),
\qquad
\beta_{ij}(u)\approx \bml_{ij}^{\T} \B(u), \quad j=1,\dots,d,
$$
with coefficient vectors $\bml_{ij}\in\mR^L$. Set
$$
\Z_t=\bigl(\wt{\B}(t/T)^{\T},\ f_{t1}\B(t/T)^{\T},\dots,f_{td}\B(t/T)^{\T}\bigr)^{\T}\in\mR^{(d+1)L},
$$
$$
\Z=(\Z_1,\dots,\Z_T)^{\T}\in\mR^{T\times (d+1)L}.
$$
Then for each $i$ we have the sieve representation
\begin{equation}
\R_{i\cdot}=\delta_i\bone_T+\Z\bml_i^0+\e_{i\cdot}+\r_{i\cdot},
\label{eq:sieve_repr}
\end{equation}
where $\R_{i\cdot}=(R_{i1},\dots,R_{iT})^{\top}$, $\e_{i\cdot}=(e_{i1},\dots,e_{iT})^{\top}$, $\bone_T=(1,\dots,1)^{\top}\in\mR^{T}$, $\bml_i^0=(\bml_{i0}^{0\top},\dots,\bml_{id}^{0\top})^{\T}$ and the remainder $\r_{i\cdot}=(r_{i1},\dots,r_{iT})^{\T}$ is given by
$$
r_{it}=\bigl[g_i(t/T)-\bml_{i0}^{0\top}\wt{\B}(t/T)\bigr]+\sum_{j=1}^d \bigl[\beta_{ij}(t/T)-\bml_{ij}^{0\top}\B(t/T)\bigr]f_{tj}.
$$
Let $\bP_Z=\Z(\Z^{\T} \Z)^{-1}\Z^{\T}$ and $\bM_Z=\I_T-\bP_Z$. The OLS sieve estimator of $\bml_i^0$ under the null restriction is
$$
\wh{\bml}_i=(\Z^{\T} \Z)^{-1}\Z^{\T} \R_{i\cdot},
$$
and the residual vector is
$$
\wh{\e}_{i\cdot}=\bM_Z \R_{i\cdot}.
$$
To estimate $\delta_i$, define
\begin{equation}
\wh \delta_i=\frac{\bone_T^{\T} \bM_Z \R_{i\cdot}}{\bone_T^{\T} \bM_Z\bone_T}.
\label{eq:deltahat}
\end{equation}
Writing $\h=\bM_Z\bone_T=(h_1,\dots,h_T)^{\T}$ and $\kappa_T=\bone_T^{\T} \bM_Z\bone_T=\sum_{t=1}^T h_t^2$, we obtain the exact decomposition
\begin{equation}
\wh\delta_i-\delta_i=\kappa_T^{-1}\sum_{t=1}^T h_t e_{it}+\kappa_T^{-1}\sum_{t=1}^T h_t r_{it}.
\label{eq:deltahat_decomp}
\end{equation}
The first term is the projected score; the second is the sieve approximation bias.

{ 
Write $\e_t=(e_{1t},\dots,e_{Nt})^{\T}$. Define
\begin{align*}
&\wb{\X}_T=T^{-1}\sum_{t=1}^T \X_t,\qquad\X_t=\e_t\eta_t,\qquad\eta_t=\frac{h_t}{\kappa_T/T}=\frac{1-\Z_t^{\T}(\Z^{\T}\Z/T)^{-1}(\Z^{\T}\bone_T/T)}{1-(\bone_T^{\T}\Z/T)(\Z^{\T}\Z/T)^{-1}(\Z^{\T}\bone_T/T)},\\
&\wb{\wk\X}_T=T^{-1}\sum_{t=1}^T \wk\X_t,\qquad\wk\X_t=\e_t\wk\eta_t,\qquad\wk\eta_t=\frac{1-\Z_t^{\T}\{\E(\Z^{\T}\Z/T)\}^{-1}\E(\Z^{\T}\bone_T/T)}{1-\E(\bone_T^{\T}\Z/T)\{\E(\Z^{\T}\Z/T)\}^{-1}\E(\Z^{\T}\bone_T/T)},
\end{align*}
and $\b_T=\kappa_T^{-1}(\sum_{t=1}^T h_t r_{1t},\dots,\sum_{t=1}^T h_t r_{Nt})^{\T}$. Then, 
\begin{equation}
\wh{\bmd}-\bmd=\wb{\wk\X}_T+(\wb{\X}_T-\wb{\wk\X}_T)+\b_T.
\qquad
\label{eq:vector_decomp}
\end{equation}
Under the following assumptions, we have
\begin{equation}
\left|\wb\X_T^{\T}\wb \X_T-\wb{\wk \X}_T^{\T}\wb{\wk \X}_T\right|=\Op\left\{T^{-3/2}L^{1/2}(\log T)\tr(\bmS)\right\},\qquad
\norm{\b_T}_2=\Op(N^{1/2}L^{-r}).
\label{eq:biasrates}
\end{equation}
}

\subsection{Sum-type test}

When the residuals are serially independent, \citet{ma2020testing} proposed the sum-type statistic
\begin{equation*}
S_{NT}
=
\frac{1}{NT}\sum_{i=1}^N
\bigl(\wh{\e}_{i\cdot}^{\top}\bone_T\bigr)^2
=
\frac{1}{NT}\sum_{i=1}^N
\Bigl(\sum_{t=1}^T \wh e_{it}\Bigr)^2.
\label{eq:SNT}
\end{equation*}
Under \(H_0\), its feasible centering and scaling constants are
\begin{equation*}
\wh\mu_{NT}
=
\frac{1}{NT}\sum_{i=1}^N\sum_{t=1}^T \wh e_{it}^2 h_t^2,
\qquad
\wh\sigma_{NT}^2
=
\frac{2}{N^2T^2}\,\widehat{\tr(\bmS^2)}
\sum_{1\le t\neq s\le T} h_t^2 h_s^2,
\label{eq:mu_sigma_sum}
\end{equation*}
where
\begin{equation*}
\widehat{\tr(\bmS^2)}
=
\frac{T^2}{\{T+(d+1)L-1\}\{T-(d+1)L\}}
\left[
\tr(\wh{\bmS}^2)
-
\frac{\tr^2(\wh{\bmS})}{T-(d+1)L}
\right],
\label{eq:trSigma2hat}
\end{equation*}
with
\[
\wh{\bmS}
=
\frac1T\sum_{t=1}^T
(\wh{\e}_t-\bar{\wh{\e}})(\wh{\e}_t-\bar{\wh{\e}})^{\T},
\qquad
\bar{\wh{\e}}=\frac1T\sum_{t=1}^T \wh{\e}_t,
\qquad
\wh{\be}_t=(\wh e_{1t},\ldots,\wh e_{Nt})^{\T}.
\]
The standardized sum statistic is then defined by
\begin{equation*}
Q_{\mathrm{SUM}}
=
\frac{S_{NT}-\wh\mu_{NT}}{\wh\sigma_{NT}}.
\label{eq:QSUM}
\end{equation*}
Under some mild conditions of \citet{ma2020testing}, they show that
\[
Q_{\mathrm{SUM}} \cd \mathcal{N}(0,1).
\]
Therefore, for a test with significance level \(\gamma\), we reject \(H_0\) whenever $Q_{\mathrm{SUM}} > z_{1-\gamma}$,
or, equivalently, whenever $p_{\mathrm{SUM}}
=1-\Phi(Q_{\mathrm{SUM}})<\gamma$.

When the error process is temporally dependent, the sum-type statistic inherits a nontrivial serial covariance structure, so the mean and variance formulas valid under serial independence are no longer appropriate. As a result, directly applying the independence-based test can cause serious size distortion. The main issue, therefore, is to characterize the null expectation and variance of the sum statistic under temporal dependence and to construct feasible estimators that yield asymptotically valid inference.

Recall the sum-type statistic
\begin{equation}
T_{\mathrm{DSUM}}=\wh{\bmd}^{\top}\wh{\bmd}=\sum_{i=1}^N \wh\delta_i^2.
\label{eq:sumstat}
\end{equation}
{ 
Under \eqref{eq:vector_decomp},
$$
T_{\mathrm{DSUM}}=\left\{\wb{\wk\X}_T+(\wb{\X}_T-\wb{\wk\X}_T)+\b_T\right\}^{\T}\left\{\wb{\wk\X}_T+(\wb{\X}_T-\wb{\wk\X}_T)+\b_T\right\}.
$$
For $h\in\mZ$, let
$$
\wk\bG_h=\cov(\wk\X_t,\wk\X_{t+h}),
\qquad
\wk\bO_T=\wk\bG_0+2\sum_{h=1}^{T-1}\Bigl(1-\frac{h}{T}\Bigr)\wk\bG_h,
\qquad
\wk\bO=\sum_{h\in\mZ}\wk\bG_h.
$$
Then under $H_0$, by \eqref{eq:biasrates},
\begin{align}\label{mean_variance}
    &\E(T_{\mathrm{DSUM}})\approx\E\left(\wb{\wk\X}_T^{\T}\wb{\wk\X}_T\right)=T^{-1}\tr(\wk\bO_T)=:\wk\mu_T,\nonumber\\
    &\var(T_{\mathrm{DSUM}})\approx\var\left(\wb{\wk\X}_T^{\T}\wb{\wk\X}_T\right)\approx 2T^{-2}\tr(\wk\bO_T^2)=:\wk\sigma_T^2.
\end{align}
}
Since the nonlinear sieve step and temporal dependence make direct feasible formulas for the null centering and scaling cumbersome, we estimate them by a circular moving block bootstrap based on the leading projected score process.

Write
$$
\wh{\X}_t=\wh{\e}_t\eta_t,
\qquad
\bar{\wh{\X}}=T^{-1}\sum_{t=1}^T \wh{\X}_t,
\qquad
\tilde{\X}_t=\wh{\X}_t-\bar{\wh{\X}},
\qquad
\wh{\e}_t=(\wh e_{1t},\dots,\wh e_{Nt})^{\T}.
$$
Let $\ell$ be the bootstrap block length and, for notational simplicity, suppose first that $T=k\ell$ for some integer $k$; the general case only adds an $O(\ell/T)$ edge term. For each starting point $j\in\{1,\dots,T\}$, form the circular block
$$
\mc B_j=(\tilde{\X}_j,\tilde{\X}_{j+1},\dots,\tilde{\X}_{j+\ell-1}),
$$
where the index is understood modulo $T$. Draw block indices $I_1,\dots,I_k$ independently and uniformly from $\{1,\dots,T\}$, concatenate the sampled blocks, and denote the resulting bootstrap series $(\mc B_{I_1},\dots,\mc B_{I_k})$ by $\{\X_t^*\}_{t=1}^T$. Let
$$
\bar{\X}_T^*=T^{-1}\sum_{t=1}^T \X_t^*,
\qquad
T_{\mathrm{DSUM}}^*=\bar{\X}_T^{*\top}\bar{\X}_T^*.
$$
We use the conditional bootstrap mean and variance
\begin{equation}
\wh\mu_{B,T}=\E^*(T_{\mathrm{DSUM}}^*),
\qquad
\wh\sigma_{B,T}^2=\var^*(T_{\mathrm{DSUM}}^*),
\label{eq:bootsumpars}
\end{equation}
as the feasible centering and scaling constants. 
% For later comparison, define the Bartlett-smoothed theoretical centering target
% \begin{equation}
% \mu_T'=T^{-1}\tr(\bO_{\ell}),
% \qquad
% \bO_{\ell}=\bG_0+\sum_{1\le |h|<\ell}\Bigl(1-\frac{|h|}{\ell}\Bigr)\bG_h.
% \label{eq:muprime}
% \end{equation}
The resulting sum-type test statistic is
\begin{equation}
Q_{\mathrm{DSUM}}=\frac{T_{\mathrm{DSUM}}-\wh\mu_{B,T}}{\wh\sigma_{B,T}}.
\label{eq:qsum}
\end{equation}
Theorem~\ref{thm:sumclt} below shows that, under the null hypothesis,
\[
Q_{\mathrm{DSUM}} \cd \mathcal{N}(0,1).
\]
Accordingly, for a test with significance level $\gamma$, we reject the null hypothesis whenever
\[
Q_{\mathrm{DSUM}} > z_{1-\gamma},
\]
where $z_{1-\gamma}$ denotes the $(1-\gamma)$-quantile of the standard normal distribution.

\subsection{Max-type test}

To motivate the max-type procedure under temporal dependence, we first briefly recall the benchmark proposed by \citet{ma2024adaptive} for the conditional time-varying factor model with serially independent errors. Define the coordinatewise squared Studentized statistics by
\begin{equation*}
t_i^2
=
\frac{\kappa_T^2\wh\delta_i^2}{T\wh\sigma_{ii}}
=
T^{-1}\wh\sigma_{ii}^{-1}
\bigl(\wh{\be}_{i\cdot}^{\top}\bone_T\bigr)^2,
\qquad
1\le i\le N,
\label{eq:ti2_ind}
\end{equation*}
where
\begin{equation}
\wh\sigma_{ij}
=
\frac{\wh{\be}_{i\cdot}^{\top}\wh{\be}_{j\cdot}}{T-d-1}.
\label{eq:sigmaij_ind}
\end{equation}
The corresponding max-type statistic is
\begin{equation*}
Q_{\mathrm{MAX}} =\max_{1\le i\le N} t_i^2.
\label{eq:max_ind}
\end{equation*}
Under the null hypothesis and serial independence, \citet{ma2024adaptive} showed that
\begin{equation}
\Pr\!\left\{
Q_{\mathrm{MAX}}-2\log N+\log(\log N)\le x
\right\}
\to
F(x)
:=
\exp\!\left\{
-\pi^{-1/2}e^{-x/2}
\right\},
\qquad x\in\mR.
\label{eq:gumbel_ind}
\end{equation}
Accordingly, the associated \(p\)-value is $p_{\mathrm{MAX}} 
=
1-
F\!\left(
Q_{\mathrm{MAX}}-2\log N+\log(\log N)
\right)$
% \begin{equation*}
% p_{\mathrm{MAX}} 
% =
% 1-
% F\!\left(
% Q_{\mathrm{MAX}}-2\log N+\log(\log N)
% \right),
% \label{eq:pmax_ind}
% \end{equation*}
and a level-\(\gamma\) test rejects \(H_0\) whenever $p_{\mathrm{MAX}} <\gamma$.
% \begin{equation*}
% p_{\mathrm{MAX}} <\gamma.
% \label{eq:reject_ind_p}
% \end{equation*}
Equivalently, one may reject \(H_0\) whenever $Q_{\mathrm{MAX}}>2\log N-\log(\log N)+q^{\mathrm{EV}}_{1-\gamma}$,
% \begin{equation*}
% Q_{\mathrm{MAX}}>2\log N-\log(\log N)+q^{\mathrm{EV}}_{1-\gamma},
% \qquad
% q^{\mathrm{EV}}_{1-\gamma}
% =
% -2\log\!\bigl\{-\sqrt{\pi}\log(1-\gamma)\bigr\},
% \label{eq:reject_ind_threshold}
% \end{equation*}
where \(q^{\mathrm{EV}}_{1-\gamma}=F^{-1}(1-\gamma)\).

However, once the residual process exhibits temporal dependence, the above independence-based construction is no longer directly valid. In particular, serial dependence changes the stochastic scale of the average-alpha estimator \(\wh\delta_i\), so the variance proxy \(\wh\sigma_{ii}\) in \eqref{eq:sigmaij_ind}, which only reflects contemporaneous variation, no longer provides an appropriate normalization. More importantly, temporal dependence also affects the extremal behavior of the coordinatewise statistics, so directly using the Gumbel calibration in \eqref{eq:gumbel_ind} may lead to non-negligible size distortion. Therefore, in the time-dependent setting, the main task is to construct a max-type statistic that is normalized by a long-run variance measure and to develop a valid calibration method that properly accounts for serial dependence.

For each coordinate $i$, define the long-run variance
\begin{equation*}
\wk\sigma_i^0=\sum_{h\in\mZ}\cov(\wk X_{it},\wk X_{i,t+h}).
\end{equation*}
Estimate it by
\begin{equation}
\wh\sigma_i=\sum_{|h|\le M}\omega\Bigl(\frac{h}{M}\Bigr)\wh\phi_{i,h},
\qquad
\wh\phi_{i,h}=(T-|h|)^{-1}\sum_{t=|h|+1}^T \wh e_{it}\wh e_{i,t-|h|}\eta_t\eta_{t-|h|},
\label{eq:sigmai}
\end{equation}
where $\omega(\cdot)$ is the Bartlett kernel,
$\omega(x)=(1-|x|)\mathbf{1}(|x|\le 1)$.
The max statistic is
\begin{equation}
Q_{\mathrm{DMAX}}=\max_{1\le i\le N} \frac{T\wh\delta_i^2}{\wh\sigma_i}.
\label{eq:maxstat}
\end{equation}
By Theorem~\ref{thm:maxgumbel}, for every fixed \(x\in\mR\),
\[
\Pr\bigl(Q_{\mathrm{DMAX}}-2\log N+\log(\log N)\le x\bigr)\to F(x)=\exp\{-\pi^{-1/2}e^{-x/2}\}.
\]
Nevertheless, the convergence to the limiting extreme-value distribution is often slow in finite samples. As a result, inference based solely on this asymptotic approximation may suffer from non-negligible finite-sample distortion.

For finite-sample calibration, we use the same circular moving block bootstrap as in the sum statistic. Applying \eqref{eq:sigmai} to the bootstrap sample $\{\X_t^*\}_{t=1}^T$ gives $\wh\sigma_i^*$, and we define
\begin{equation}
Q_{\mathrm{DMAX}}^*=\max_{1\le i\le N}\frac{T(\bar X_{i,T}^*)^2}{\wh\sigma_i^*},
\label{eq:qmaxboot}
\end{equation}
where $\bar X_{i,T}^*$ is the $i$th component of $\bar{\X}_T^*$. The associated bootstrap p-value is
$$
p_{\mathrm{DMAX}}^{\mathrm{boot}}=\Pr^*(Q_{\mathrm{DMAX}}^*>Q_{\mathrm{DMAX}}).
$$
The extreme-value approximation remains useful analytically, and we retain it in the main null theory below.

\subsection{Cauchy combination test}

The sum-type and max-type procedures are designed for different classes of alternatives. 
Roughly speaking, the sum-type statistic is more powerful under relatively dense alternatives, 
where many components of the alpha vector deviate from zero, whereas the max-type statistic is 
more sensitive to sparse alternatives, where only a small number of components are nonzero. 
In practice, however, the sparsity level of the alternative is typically unknown. Therefore, 
relying on a single test may lead to a substantial loss of power when the underlying alternative 
structure is misspecified.

For the serially independent conditional time-varying factor model, \citet{ma2024adaptive} 
combined the sum-type and max-type procedures to obtain an adaptive test that performs well 
across a broad range of alternatives. Motivated by the same idea, we construct a Cauchy 
combination test in the present time-dependent setting by integrating the feasible \(p\)-values 
from the dependent sum-type and max-type statistics developed above.

Recall that the feasible \(p\)-values are defined by
\begin{equation}
p_{\mathrm{DSUM}}=1-\Phi(Q_{\mathrm{DSUM}}),
\qquad
p_{\mathrm{DMAX}}^{\mathrm{boot}}
=\Pr^*\bigl(Q_{\mathrm{DMAX}}^*>Q_{\mathrm{DMAX}}\bigr).
\end{equation}
We then define the Cauchy combination statistic by
\begin{equation}
T_{\mathrm{CC}}
=
\frac{1}{2}\tan\!\bigl\{\pi(1/2-p_{\mathrm{DSUM}})\bigr\}
+
\frac{1}{2}\tan\!\bigl\{\pi(1/2-p_{\mathrm{DMAX}}^{\mathrm{boot}})\bigr\}.
\label{eq:ccstat}
\end{equation}
The corresponding combined \(p\)-value is
\begin{equation}
p_{\mathrm{CC}}
=
1-G(T_{\mathrm{CC}}),
\label{eq:pcc}
\end{equation}
where \(G(x)=1/2+\pi^{-1}\arctan(x)\) is the standard Cauchy distribution function.

The validity of this combination rule relies on the asymptotic independence between the 
sum-type and max-type statistics established below. Consequently, under the null hypothesis, 
\(p_{\mathrm{CC}}\) is asymptotically valid and yields an adaptive test that remains effective 
without prior knowledge of the sparsity level. For a test with significance level \(\gamma\), 
we reject \(H_0\) whenever $p_{\mathrm{CC}}<\gamma$, or, equivalently, whenever $T_{\mathrm{CC}}>G^{-1}(1-\gamma)$.

\section{Theoretical Results}
\label{sec:main}

\subsection{Assumptions}

The assumptions below are a synthesis of the conditions used in \citet{ma2020testing}, \citet{ma2024adaptive}, and \citet{ma2024dependent}. They are stated directly in the form needed for the present paper.

\begin{assumption}[Smooth coefficient functions]
\label{ass:smooth}
Let $\mathcal{H}_r$ denote the collection of all functions on $[0,1]$ such that the $m$-th order derivative satisfies the H\"older condition of order $n$ with $r=m+n$, i.e., there exists a constant $ C\in(0,\infty)$ such that for each $f\in\mathcal{H}_r$, $|f^{(m)}(x_1)-f^{(m)}(x_2)|\le C|x_1-x_2|^n$, for any $0\le x_1,x_2\le 1$. For each $i\le N$ and $j\le d$, the functions $g_i(\cdot), \beta_{ij}(\cdot)\in \mathcal{H}_r$ for some $r>3/2$.
\end{assumption}

\begin{assumption}[Factors]
\label{ass:factors}
The sequence $\{\f_t\}_{t=1}^T$ is strictly stationary, independent of $\{\e_t\}_{t=1}^T$, and strong mixing with mixing coefficient $\alpha_f(\cdot)$ satisfying $\sum_{k=0}^\infty\alpha_f(k)^{\kappa/(2+\kappa)}<\infty$ for some $\kappa>0$. $\E||\bm{f}_t||^{4(2+\kappa)}<\infty$, and $0<c\leq \lambda_{\min}(\E\{(1,\f_t^{\T})^{\T}(1,\f_t^{\T})\})\leq \lambda_{\max}(\E\{(1,\f_t^{\T})^{\T}(1,\f_t^{\T})\})\leq C<\infty$ uniformly for $t$.
\end{assumption}

\begin{assumption}[Dependent idiosyncratic errors]
\label{ass:errors}
The sequence $\{\e_t\}_{t=1}^T$ is strictly stationary with mean zero and follows the linear-process representation
$$
\e_t=\bmS^{1/2}\sum_{k=0}^\infty b_k \z_{t-k},
$$
where $\sum_{k=0}^\infty |b_k|<\infty$, $\sum_{k=0}^\infty b_k=s\neq 0$, and $b_k=\mathrm{o}(k^{-5-\eta_b})$ for some $\eta_b>0$. The innovations $\z_t=(z_{1t},\dots,z_{Nt})^{\T}$ have independent coordinates across both $i$ and $t$, satisfy $\E z_{it}=0$, $\E z_{it}^2=1$, and $\sup_{i,t}\E|z_{it}|^q<\infty$ for some $q>8$.
\end{assumption}

\begin{assumption}[Cross-sectional covariance regularity]
\label{ass:sigma}
{ 
$\tr(\bmS)\asymp N$ and 
$$
\log^2N\le \frac{\tr^2(\bmS)}{\tr(\bmS^2)}\ll T^{\frac{2r-1}{2r+1}}(\log T)^{-\frac{4r}{2r+1}}.
$$}
Further, if $\R_0=\D_0^{-1/2}\bO \D_0^{-1/2}$ with $\D_0=\diag(\bO)$ and $\bO=\sum_{h\in\mZ}\cov(\e_t,\e_{t+h})$, then
$$
\max_{1\le i<j\le N}|[\R_0]_{ij}|\le r_0<1,
\qquad
\max_{1\le j\le N}\sum_{i=1}^N [\R_0]_{ij}^2\le C_R<\infty.
$$
\end{assumption}

\begin{assumption}[Growth rates]
\label{ass:growth}
{ 
As $(N,T)\to\infty$,
\begin{align*}
    \left\{\frac{T\tr^2(\bmS)}{\tr(\bmS^2)}\right\}^{\frac{1}{2r}}\ll L \ll \frac{T\tr(\bmS^2)}{\log^2 T\tr^2(\bmS)},
\end{align*}
and the bandwidth $M=M_{N,T}$ used in the long-run variance estimators and the block length $\ell=\ell_{N,T}$ used in the circular moving block bootstrap satisfy
\begin{align*}
    \frac{\tr(\bmS)}{\tr^{1/2}(\bmS^2)}\ll \ell \le T^{1/2}L^{-1/2}\log T,\qquad \log N\ll M\ll \min\left(\frac{T^{1/2}}{\log N}, \frac{TL^{-1}}{\log N},\frac{TL^{-1/2}}{\log T \log N}\right).
\end{align*}}
\end{assumption}

\begin{remark}
Assumptions \ref{ass:smooth}--\ref{ass:growth} are a consolidated version of the regularity conditions used in \citet{ma2020testing}, \citet{ma2024adaptive}, and \citet{ma2024dependent}, rewritten in a form convenient for the present time-dependent conditional factor model.

Assumption \ref{ass:smooth} is the same smoothness requirement as Assumption 2.1 in \citet{ma2020testing} and Assumption 2.1 in \citet{ma2024adaptive}. It requires the conditional intercept and loading functions to lie in a common H\"older class, so that the sieve approximation error is uniformly of order \(L^{-r}\). The displayed bound on \(r_{it}\) is the standard spline-approximation consequence of this smoothness condition, after allowing the approximation error to scale with the magnitude of the observed factors.

% Assumption \ref{ass:design} corresponds to the sieve-design regularity imposed in Assumption 2.2 of \citet{ma2020testing} and \citet{ma2024adaptive}. Its role is twofold. First, quasi-uniform knots and the eigenvalue bounds for \(T^{-1}\Z^{\T} \Z\) guarantee that the spline regression is well conditioned and that no basis direction becomes asymptotically degenerate. Second, the requirement \(\kappa_T=\bone_T^{\T} \bM_Z\bone_T=T\omega_T\) with \(\omega_T\) bounded away from zero and infinity ensures that the intercept direction is not asymptotically absorbed by the sieve space, so the alpha component remains identifiable after partialling out the varying-coefficient terms.

Assumption \ref{ass:factors} is aligned with the factor-side conditions in Assumption 2.2(ii)--(iv) of \citet{ma2020testing} and \citet{ma2024adaptive}, and also with Condition (C1) of \citet{ma2024dependent}. It requires a fixed-dimensional, stationary, and nondegenerate factor process that is independent of the idiosyncratic errors. The boundedness condition imposed here is slightly stronger than the moment-and-mixing formulation in the earlier conditional-factor papers, but it serves the same purpose: namely, to control the stochastic order of the projection terms uniformly in \(t\) and hence simplify the feasible-oracle comparison.

Assumption \ref{ass:errors} is the dependent-data analogue of the i.i.d.\ error condition in Assumption 2.3 of \citet{ma2020testing} and \citet{ma2024adaptive}, and is directly patterned after Condition (C2) of \citet{ma2024dependent}. It allows serial dependence through an absolutely summable linear process while preserving cross-sectional innovation independence. Relative to \citet{ma2024dependent}, the present formulation strengthens the moment requirement from a finite fourth moment to a finite \(q\)-th moment with \(q>8\), and replaces the generic \(o(k^{-5})\) coefficient decay by \(o(k^{-5-\eta_b})\). These stronger versions are imposed only to make the higher-order uniform bounds and the subsequent resampling arguments transparent.

Assumption \ref{ass:sigma} combines the cross-sectional covariance restrictions used for the conditional-factor sum statistic and those used for the dependent max statistic. In particular, the spectral-norm and average-variance bounds are the counterparts of the covariance regularity in \citet{ma2024dependent}, while the restrictions on \(\mathbf R_0\) play the same role as the bounded-correlation conditions in Assumption 2.4(ii)--(iii) and Assumptions 3.1--3.2 of \citet{ma2024adaptive}. Econometrically, this assumption excludes excessively strong cross-sectional dependence, guarantees that the long-run covariance is nondegenerate on average, and provides the weak-dependence structure needed for both the Gaussian approximation of the sum-type statistic and the Gumbel approximation of the max-type statistic.

Finally, Assumption \ref{ass:growth} collects the rate conditions scattered across the preceding papers into one place. The restrictions on $L$ are not new qualitative assumptions, but convenient sufficient conditions ensuring that the spline approximation bias and the uniform estimation error are negligible. Likewise, the block-length conditions on \(\ell\) are bookkeeping assumptions introduced only to guarantee that the bootstrap approximation respects the same dependence range as the underlying linear process.
\end{remark}

We now state the main theorems. 
% The first result gives a rate-explicit Bahadur-type expansion for the estimated average alphas.

% \begin{theorem}[Uniform expansion of estimated average alphas]
% \label{thm:expansion}
% Under Assumptions \ref{ass:smooth}--\ref{ass:errors},
% \begin{equation*}
% \norm{\wh{\bmd}-\bmd-\wb{\X}_T}_\infty=O_p\Bigl(L^{-r}\cancel{+\frac{L\log N}{T}}\Bigr),
% \qquad
% \norm{\wh{\bmd}-\bmd-\wb{\X}_T}_2=O_p\Bigl(N^{1/2}L^{-r}\cancel{+\frac{N^{1/2}L\log N}{T}}\Bigr).
% \label{eq:uniformexp2}
% \end{equation*}
% \end{theorem}
% \begin{remark}
% The extra term $L\log N/T$ is the price of replacing the ideal projection by the estimated spline projection. Under Assumption \ref{ass:growth}, both terms in \eqref{eq:uniformexp} are $o(T^{-1/2})$.
% \end{remark}
\begin{theorem}
\label{thm:sumclt}
Suppose Assumptions \ref{ass:smooth}--\ref{ass:growth} hold and $H_0$ is true. Then
\begin{equation*}
\frac{T_{\mathrm{DSUM}}-\wk \mu_T}{\wk \sigma_T}\cd \mathcal{N}(0,1).
\label{eq:sumclt}
\end{equation*}
Moreover, the bootstrap centering and scaling errors satisfy
\begin{equation*}
\frac{|\wh\mu_{B,T}-\wk\mu_T|}{\wk\sigma_T}=o_p(1),\qquad
\Bigl|\frac{\wh\sigma_{B,T}^2}{\wk\sigma_T^2}-1\Bigr|=o_p(1).
\end{equation*}
Hence
\begin{equation*}
Q_{\mathrm{DSUM}}=\frac{T_{\mathrm{DSUM}}-\wh\mu_{B,T}}{\wh\sigma_{B,T}}\cd \mathcal{N}(0,1).
\label{eq:qsumclt}
\end{equation*}
\end{theorem}

Theorem \ref{thm:sumclt} establishes the asymptotic normality of the sum statistic and provides a theoretical justification for the bootstrap centering and scaling.

\begin{theorem}
\label{thm:sumpower}
Under Assumptions \ref{ass:smooth}--\ref{ass:growth} and the following alternatives
\begin{align*}
\norm{\bmd}_2=O\{T^{-1}\tr^{1/4}(\wk\bO_T^2)\},
\end{align*}
we have
\begin{equation*}
Q_{\mathrm{DSUM}}=\frac{\bmd^{\T}\bmd}{\sqrt{2T^{-2}\tr(\wk\bO_T^2)}}+Z_T+\op(1),
\label{eq:sumpowerexpansion}
\end{equation*}
where $Z_T\cd \mathcal{N}(0,1)$. Hence the asymptotic power at level $\gamma$ is
$$
\beta_{\mathrm{DSUM}}(\bmd)=\Phi\Bigl(-z_\gamma+\frac{\bmd^{\T}\bmd}{\sqrt{2T^{-2}\tr(\wk\bO_T^2)}}\Bigr)+o(1).
$$
\end{theorem}

Theorem \ref{thm:sumpower} characterizes the power properties of the sum statistic.

\begin{theorem}
\label{thm:maxgumbel}
Under Assumptions \ref{ass:smooth}--\ref{ass:growth} and $H_0$, for every fixed $x\in\mR$,
\begin{equation*}
\Pr\bigl(Q_{\mathrm{DMAX}}\le x\bigr)\to \exp\{-\pi^{-1/2}e^{-x/2}\}=:F(x). 
\end{equation*}
\end{theorem}

Theorem \ref{thm:maxgumbel} establishes the limiting Gumbel distribution of the max statistic under the null hypothesis. Next, we establish the Block bootstrap validity for the max statistic.

\begin{theorem}
\label{thm:maxboot}
Under Assumptions \ref{ass:smooth}--\ref{ass:growth} and $H_0$, the circular moving block bootstrap satisfies
\begin{equation*}
\sup_{x\in\mR}\Bigl|\Pr^*(Q_{\mathrm{DMAX}}^*\le x)-F(x)\Bigr|\cp 0,
\label{eq:maxbootgumbel}
\end{equation*}
and the bootstrap p-value satisfies
\begin{equation*}
\Pr\bigl(p_{\mathrm{DMAX}}^{\mathrm{boot}}\le \gamma\bigr)\to \gamma,
\qquad 0<\gamma<1.
\label{eq:maxbootp}
\end{equation*}
\end{theorem}

Next, we establish the power properties of the max statistic.
\begin{theorem}
\label{thm:maxpower}
Suppose the assumptions of Theorem \ref{thm:maxgumbel} hold. Let
$$
\mc A(c)=\Bigl\{\bmd\in\mR^N:\ \max_{1\le i\le N}\frac{|\delta_i|}{\sqrt{\wk\sigma_i^0}}\ge c\sqrt{\frac{\log N}{T}}\Bigr\}.
$$
If $c>2$, then
\begin{equation*}
\inf_{\bmd\in\mc A(c)} \Pr\bigl(Q_{\mathrm{DMAX}}>q_{1-\gamma}^{\mathrm{EV}}\bigr)\to 1,
\label{eq:maxpower}
\end{equation*}
where $q_{1-\gamma}^{\mathrm{EV}}$ is the $(1-\gamma)$ quantile of $F$.
% More explicitly, for any $\bmd\in\mc A(c)$,
% \begin{equation}
% \Pr\bigl(Q_{\mathrm{DMAX}}>q_{1-\gamma}^{\mathrm{EV}}\bigr)
% \ge 1-C\exp\Bigl[-\frac{1}{2}(c-2)^2\log N\Bigr]-o(1).
% \label{eq:maxpowerexp}
% \end{equation}
\end{theorem}

\begin{theorem}
\label{thm:indep}
Under Assumptions \ref{ass:smooth}--\ref{ass:growth} and $H_0$, for every fixed $(x,y)\in\mR^2$,
\begin{equation*}
\Pr\bigl(Q_{\mathrm{DSUM}}\le x,\ Q_{\mathrm{DMAX}}\le y\bigr)
\to \Phi(x)F(y).
\label{eq:indep}
\end{equation*}
The same conclusion remains valid under local alternatives satisfying
\begin{equation*}
\norm{\bmd}_0=o\Bigl(\frac{N}{\log^2\log N}\Bigr),
\qquad
\norm{\bmd}_2=O\left\{T^{-1/2}\tr^{1/4}(\wk\bO_T^2)\right\},
\label{eq:localalt}
\end{equation*}
provided that the remainder bounds in Theorems \ref{thm:sumpower} and \ref{thm:maxgumbel} still vanish.
\end{theorem}

Theorem \ref{thm:indep} establishes the asymptotic independence between the sum and max statistics, which in turn provides the foundation for Theorem \ref{thm:cc}, which establishes the validity of the Cauchy combination test.

\begin{theorem}
\label{thm:cc}
Under the assumptions of Theorem \ref{thm:indep} and $H_0$,
\begin{equation*}
\Pr\bigl(p_{\mathrm{DCC}}\le \gamma\bigr)\to \gamma, \qquad 0<\gamma<1.
\label{eq:ccvalid}
\end{equation*}
Moreover, if either the sum statistic or the max statistic is consistent under an alternative sequence, then so is the Cauchy combination test. In particular,
\begin{equation*}
\Pr\bigl(p_{\mathrm{DCC}}\le \gamma\bigr)
\ge \max\Bigl\{\Pr\bigl(p_{\mathrm{DSUM}}\le \gamma/2\bigr),\ \Pr\bigl(p_{\mathrm{DMAX}}\le \gamma/2\bigr)\Bigr\}.
\label{eq:ccpower}
\end{equation*}
\end{theorem}

Moreover, under an alternative sequence along which the asymptotic independence in Theorem~\ref{thm:indep} continues to hold, the DCC test satisfies the sharper lower bound
\begin{equation*}
\label{eq:ccpower_ie}
\Pr\bigl(p_{\mathrm{DCC}}\le \gamma\bigr)
\ge
\Pr\Bigl(p_{\mathrm{DSUM}}\le \frac{\gamma}{2}\Bigr)
+
\Pr\Bigl(p_{\mathrm{DMAX}}\le \frac{\gamma}{2}\Bigr)
-
\Pr\Bigl(p_{\mathrm{DSUM}}\le \frac{\gamma}{2},\,
p_{\mathrm{DMAX}}\le \frac{\gamma}{2}\Bigr)
+o(1).
\end{equation*}
Therefore, whenever both component tests have nontrivial power, the DCC procedure inherits the power of their union and can be more powerful than either component test alone.

\section{Simulation}
\label{sec:simulation}

We conduct Monte Carlo experiments to assess the finite-sample performance of the proposed testing procedures in conditional time-varying factor models. Throughout, we report results for six methods: the original sum-type, max-type, and Cauchy-combination tests, denoted by \textrm{SUM}, \textrm{MAX}, and \textrm{CC}, together with their dependence-robust bootstrap counterparts, denoted by \textrm{DSUM}, \textrm{DMAX}, and \textrm{DCC}. The bootstrap methods are calibrated by a moving block bootstrap.

\subsection{Data-generating mechanisms}
\label{subsec:sim_dgp}

We consider two examples. In both examples, the sample is generated over the periods $t=-24,\ldots,0,1,\ldots,T$, and the first 25 observations are discarded to mitigate start-up effects. The B-spline approximation uses cubic splines with total basis dimension $L=5$.

\paragraph{Example 1.}
The first design is a conditional CAPM with one time-varying factor loading:
\[
Y_{it}=\alpha_{it}+\beta_i(t/T)f_t+e_{it}, \qquad i=1,\ldots,N,\quad t=1,\ldots,T.
\]
The factor process follows an AR(1)-GARCH(1,1) recursion,
\[
f_t=0.34+0.05(f_{t-1}-0.34)+\sqrt{h_t}\,\varepsilon_t^{(f)},
\qquad
h_t=0.32+0.67h_{t-1}+0.13\xi_{t-1}^2,
\]
where $\varepsilon_t^{(f)}$ and $\xi_t$ are generated from independent standard normal draws. The conditional beta path is common across $i$ and is given by the logistic curve
\[
\beta_i(t/T)=\left\{1+\exp\bigl[-2(10t/T-2)\bigr]\right\}^{-1},
\qquad i=1,\ldots,N.
\]

\paragraph{Example 2.}
The second design is a conditional three-factor model:
\[
Y_{it}=\alpha_{it}+\sum_{j=1}^3 \beta_{ij}(t/T)f_{jt}+e_{it},
\qquad i=1,\ldots,N,\quad t=1,\ldots,T.
\]
The three factor processes are generated from AR(1)-GARCH(1,1) recursions,
\[
f_{jt}=\mu_j+\varphi_j(f_{j,t-1}-\mu_j)+\sqrt{h_{jt}}\,\varepsilon_{jt}^{(f)},
\qquad j=1,2,3,
\]
with
\[
(\mu_1,\mu_2,\mu_3)=(0.34,\,0.04,\,0.06),\qquad
(\varphi_1,\varphi_2,\varphi_3)=(0.05,\,0.07,\,0.04),
\]
and
\[
h_{jt}=a_j+b_j h_{j,t-1}+c_j \xi_{j,t-1}^2,\qquad j=1,2,3,
\]
where
\[
(a_1,a_2,a_3)=(0.32,\,0.33,\,0.26),\qquad
(b_1,b_2,b_3)=(0.67,\,0.51,\,0.72),\qquad
(c_1,c_2,c_3)=(0.13,\,0.03,\,0.05).
\]
The conditional beta paths are
\[
\beta_{i1}(t/T)=0.5+0.5z_t,\qquad
\beta_{i2}(t/T)=0.5+0.1z_t,\qquad
\beta_{i3}(t/T)=0.5+0.2z_t,
\]
where
\[
z_t=\left\{1+\exp\bigl[-2(10t/T-2)\bigr]\right\}^{-1}.
\]

The idiosyncratic error process is generated in the same way for both examples. Let
\[
\e_t=\z_t+\sum_{h=1}^{M}\A_h\z_{t-h},
\qquad
\z_t=\bmS^{1/2}\u_t,
\]
where $\u_t=(u_{1t},\ldots,u_{Nt})^{\T}$ has i.i.d.\ coordinates. We consider two innovation distributions:
\[
u_{it}\sim \mathcal{N}(0,1),
\qquad\text{or}\qquad
u_{it}\sim t(6)/\sqrt{6/4}.
\]
Hence, in both cases, the marginal variance is normalized to one.

The cross-sectional covariance matrix $\bmS=(\sigma_{ij})_{1\le i,j\le N}$ is defined by
\[
\sigma_{ii}=1,\qquad
\sigma_{ij}=
\frac{\phi_2}{|i-j|^2}\, \mathbf{1}\!\left(1\le |i-j|\le \omega N\right),
\qquad i\neq j,
\]
with
\[
(\omega,\phi_1,\phi_2)=(0.9,\,0.6,\,0.4).
\]
For the temporal dependence matrices $\A_h$, we set
\[
[\A_h]_{ii}=\frac{\phi_1}{h},\qquad
[\A_h]_{ij}=\frac{\phi_1}{h|i-j|^2}\,\mathbf{1}\!\left(1\le |i-j|\le \omega N\right),
\qquad h=1,2,
\]
and for $h\ge 3$ we take
\[
\A_h=e^{-2h}\I_N.
\]
We consider three dependence regimes:
\[
M=0,\qquad M=2,\qquad M=T-1.
\]
These correspond, respectively, to independent errors, short-range dependence, and a long-range dependent case.

\subsection{Null experiments}
\label{subsec:sim_h0}

Under the null hypothesis, we set
\[
H_0:\quad \alpha_{it}\equiv 0,\qquad i=1,\ldots,N,\quad t=1,\ldots,T.
\]
We consider
\[
(T,N)\in\{(200,250),(400,250),(200,500),(400,500)\},
\]
under both Gaussian and standardized $t(6)$ innovations, and under all three dependence settings $M\in\{0,2,T-1\}$. For the size experiments, each setting is based on 1000 Monte Carlo replications at the nominal level $5\%$. For the bootstrap procedures, we use 500 moving block bootstrap replications. 

For the moving block bootstrap, the block length is selected in a data-driven manner from the centered residual matrix under the null model. Specifically, let
$
\mathbf E_0=\mathbf{\widehat E}_0-\bone_T\bar{\mathbf e}_0^{\T}
$
denote the centered null residual matrix, where \(\mathbf E_0\in\mathbb R^{T\times N}\). For each cross-sectional series, we apply the PWSD automatic block-length selector implemented by the \texttt{pwsd()} function in the R package \texttt{blocklength} \citep{politis1995spectral,politis2004automatic,patton2009correction}, and extract the circular-bootstrap recommendation \(b_{i,\mathrm{Circ}}\), \(i=1,\dots,N\). The final bootstrap block length is then chosen as
$$
\ell
=
\max\!\left\{
2,\,
\min\!\left(
\lfloor \sqrt{T}\rfloor,\,
\left\lceil 1.5\,\mathrm{median}\{b_{1,\mathrm{Circ}},\dots,b_{N,\mathrm{Circ}}\}\right\rceil
\right)
\right\}.
$$

\begin{table}[!htbp]
\centering
\caption{Empirical sizes (\%) for Example~1 under $H_0$. The nominal level is $5\%$.}
\label{tab:size_preliminary_ex1}
\setlength{\tabcolsep}{5pt}
\renewcommand{\arraystretch}{1.05}
\begin{tabular}{lcccccccc}
\toprule
$M$ & $T$ & $N$ & \textrm{SUM} & \textrm{MAX} & \textrm{CC} & \textrm{DSUM} & \textrm{DMAX} & \textrm{DCC} \\
\midrule
\multicolumn{9}{c}{Gaussian errors} \\
\midrule
0     & 200 & 250 & 2.2   & 1.5 & 2.2 & 1.3 & 3.9 & 2.8 \\
0     & 400 & 250 & 2.6   & 1.2 & 2.0 & 2.8 & 4.9 & 3.9 \\
0     & 200 & 500 & 3.1   & 1.6 & 2.1 & 0.2 & 4.2 & 2.3 \\
0     & 400 & 500 & 1.5   & 1.3 & 1.5 & 0.7 & 4.9 & 2.7 \\
\midrule
2     & 200 & 250 & 99.9  & 99.9 & 100.0 & 5.6 & 5.2 & 4.9 \\
2     & 400 & 250 & 100.0 & 99.9 & 100.0 & 5.7 & 4.4 & 5.4 \\
2     & 200 & 500 & 100.0 & 100.0 & 100.0 & 2.8 & 5.6 & 4.4 \\
2     & 400 & 500 & 100.0 & 100.0 & 100.0 & 1.7 & 4.5 & 3.2 \\
\midrule
$T-1$ & 200 & 250 & 100.0 & 99.6 & 100.0 & 5.5 & 6.1 & 5.7 \\
$T-1$ & 400 & 250 & 100.0 & 99.9 & 100.0 & 4.5 & 5.3 & 5.3 \\
$T-1$ & 200  & 500 & 100.0 & 100.0 & 100.0 & 1.7 & 7.0 & 4.6 \\
$T-1$ & 400  & 500 & 100.0 & 100.0 & 100.0 & 0.5 & 5.8 & 2.9 \\
\midrule
\multicolumn{9}{c}{$t(6)$ errors} \\
\midrule
0     & 200 & 250 & 2.0   & 1.2 & 1.6 & 0.6 & 3.7 & 2.1 \\
0     & 400 & 250 & 1.9   & 0.8 & 1.8 & 2.2 & 4.0 & 2.7 \\
0     & 200 & 500 & 2.0   & 0.6 & 1.7 & 0.5 & 3.1 & 2.2 \\
0     & 400 & 500 & 1.5   & 1.2 & 0.9 & 0.6 & 4.9 & 2.3 \\
\midrule
2     & 200 & 250 & 99.9  & 99.7 & 99.9  & 4.5 & 4.9 & 5.0 \\
2     & 400 & 250 & 100.0 & 100.0 & 100.0 & 5.9 & 4.9 & 5.7 \\
2     & 200 & 500 & 100.0 & 100.0 & 100.0 & 3.7 & 5.2 & 4.6 \\
2     & 400 & 500 & 100.0 & 100.0 & 100.0 & 3.3 & 5.5 & 4.7 \\
\midrule
$T-1$ & 200 & 250 & 99.9  & 99.6 & 100.0 & 5.0 & 5.9 & 5.6 \\
$T-1$ & 400  & 250 & 100.0 & 99.8 & 100.0 & 5.3 & 4.4 & 4.7 \\
$T-1$ & 200  & 500 & 100.0 & 100.0 & 100.0 & 3.3 & 6.8 & 5.8 \\
$T-1$ & 400  & 500 & 100.0 & 100.0 & 100.0 & 2.4 & 6.0 & 4.9 \\
\bottomrule
\end{tabular}
\end{table}

\begin{table}[!htbp]
\centering
\caption{Empirical sizes (\%) for Example~2 under $H_0$. The nominal level is $5\%$.}
\label{tab:size_preliminary_ex2}
\setlength{\tabcolsep}{5pt}
\renewcommand{\arraystretch}{1.05}
\begin{tabular}{lcccccccc}
\toprule
$M$ & $T$ & $N$ & \textrm{SUM} & \textrm{MAX} & \textrm{CC} & \textrm{DSUM} & \textrm{DMAX} & \textrm{DCC} \\
\midrule
\multicolumn{9}{c}{Gaussian errors} \\
\midrule
0     & 200 & 250 & 2.1   & 0.8  & 1.3  & 0.8 & 5.6 & 3.5 \\
0     & 400 & 250 & 1.7   & 1.0  & 1.2  & 0.9 & 5.2 & 3.2 \\
0     & 200 & 500 & 2.8   & 1.0  & 2.1  & 0.4 & 5.8 & 2.6 \\
0     & 400 & 500 & 2.2   & 0.9  & 1.8  & 0.5 & 5.3 & 2.4 \\
\midrule
2     & 200 & 250 & 99.9  & 98.4 & 100.0 & 6.8 & 6.4 & 6.1 \\
2     & 400 & 250 & 100.0 & 98.7 & 100.0 & 6.1 & 5.9 & 5.7 \\
2     & 200 & 500 & 100.0 & 99.8 & 100.0 & 3.5 & 5.6 & 3.9 \\
2     & 400 & 500 & 100.0 & 99.9 & 100.0 & 2.8 & 5.2 & 3.1 \\
\midrule
$T-1$ & 200 & 250 & 100.0 & 98.8 & 99.4  & 7.1 & 5.8 & 6.4 \\
$T-1$ & 400 & 250 & 100.0 & 99.1 & 99.5  & 6.2 & 5.4 & 5.8 \\
$T-1$ & 200 & 500 & 100.0 & 99.0 & 100.0 & 5.2 & 5.3 & 3.4 \\
$T-1$ & 400 & 500 & 100.0 & 99.2 & 100.0 & 4.5 & 4.9 & 2.8 \\
\midrule
\multicolumn{9}{c}{$t(6)$ errors} \\
\midrule
0     & 200 & 250 & 0.8   & 0.4  & 0.5  & 0.3 & 4.7 & 2.7 \\
0     & 400 & 250 & 0.6   & 0.5  & 0.6  & 0.4 & 4.5 & 2.5 \\
0     & 200 & 500 & 1.3   & 1.0  & 0.7  & 0.2 & 4.9 & 2.3 \\
0     & 400 & 500 & 1.0   & 1.2  & 0.5  & 0.3 & 4.7 & 2.2 \\
\midrule
2     & 200 & 250 & 100.0 & 99.3 & 100.0 & 5.7 & 5.8 & 5.3 \\
2     & 400 & 250 & 100.0 & 99.6 & 100.0 & 5.1 & 5.4 & 4.9 \\
2     & 200 & 500 & 100.0 & 99.9 & 100.0 & 3.2 & 5.5 & 4.4 \\
2     & 400 & 500 & 100.0 & 100.0 & 100.0 & 2.7 & 5.1 & 3.8 \\
\midrule
$T-1$ & 200 & 250 & 99.9  & 98.1 & 100.0 & 4.3 & 4.8 & 3.9 \\
$T-1$ & 400 & 250 & 100.0 & 98.6 & 100.0 & 3.7 & 4.5 & 3.4 \\
$T-1$ & 200 & 500 & 100.0 & 99.9 & 100.0 & 3.0 & 5.1 & 4.7 \\
$T-1$ & 400 & 500 & 100.0 & 100.0 & 100.0 & 2.5 & 4.8 & 4.2 \\
\bottomrule
\end{tabular}
\end{table}

Table~\ref{tab:size_preliminary_ex1} and \ref{tab:size_preliminary_ex2} reports the currently available preliminary size results for Example~1 and 2, respectively. 
All the results conveys two main messages. First, in the independent case $M=0$, all six procedures are somewhat conservative, with rejection frequencies generally below the nominal level. Among the bootstrap procedures, \textrm{DMAX} is typically closest to the target level, whereas \textrm{DSUM} is more conservative.

Second, once temporal dependence is introduced, the original procedures \textrm{SUM}, \textrm{MAX}, and \textrm{CC} become severely oversized, with empirical sizes essentially equal to one in the currently available designs. This pattern is entirely consistent with the fact that their calibration ignores serial dependence. In contrast, the bootstrap procedures \textrm{DSUM}, \textrm{DMAX}, and \textrm{DCC} remain much more stable, with empirical sizes mostly lying in a moderate range around the nominal level. Overall, the preliminary results strongly support the use of block-bootstrap calibration in the presence of serially dependent errors.

\subsection{Alternative experiments}
\label{subsec:sim_h1}

Under the alternative hypothesis, we generate sparse and time-varying alpha signals. Let
\[
S\subset \{1,\ldots,N\}, \qquad |S|=s,
\]
where the support $S$ is drawn uniformly at random. For $i\in S$, the baseline signal magnitude is
\[
a_i=
\sqrt{\frac{c_M\log N}{sT}},
\]
where
\[
c_M=
\begin{cases}
12, & M=0,\\
80, & M=2,\\
90, & M=T-1.
\end{cases}
\]
In the current implementation, the nonzero alpha path is time-varying and takes the form
\[
\alpha_{it}=a_i+0.35|a_i|\,g_i(t/T), \qquad i\in S,
\]
where $g_i(\cdot)$ is a centered and normalized cosine component with a random frequency drawn from $[\pi/2,\pi]$. For $i\notin S$, we set $\alpha_{it}=0$. This design preserves the sparsity pattern while allowing for moderate time variation in the nonzero alpha trajectories.

For the power experiments, we currently fix
\[
(T,N)=(200,250),
\]
and vary the sparsity level over
\[
s\in\{1,5,9,\ldots,101\}.
\]
The same two innovation distributions and the same three dependence settings are considered. The current power plots are preliminary and are based on 1000 Monte Carlo replications and 500 bootstrap replications. In the final version, we will report the results for both Example~1 and Example~2 under a larger number of replications.

\begin{figure}[!htbp]
\centering
\includegraphics[width=\textwidth]{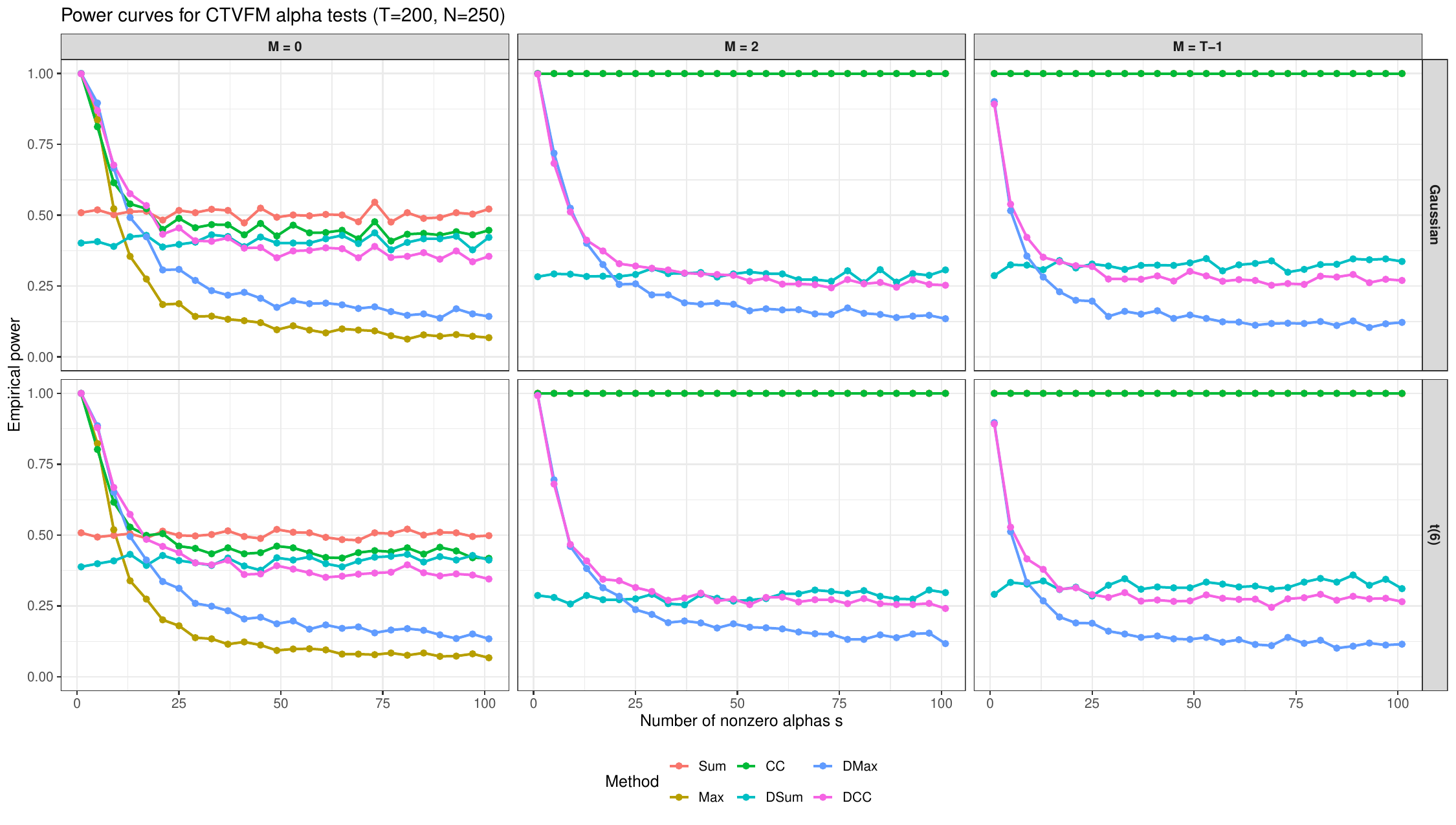}
\caption{Preliminary power curves for Example~1 under different sparse alternatives.}
\label{fig:power_preliminary_ex1}
\end{figure}

\begin{figure}[!htbp]
\centering
\includegraphics[width=\textwidth]{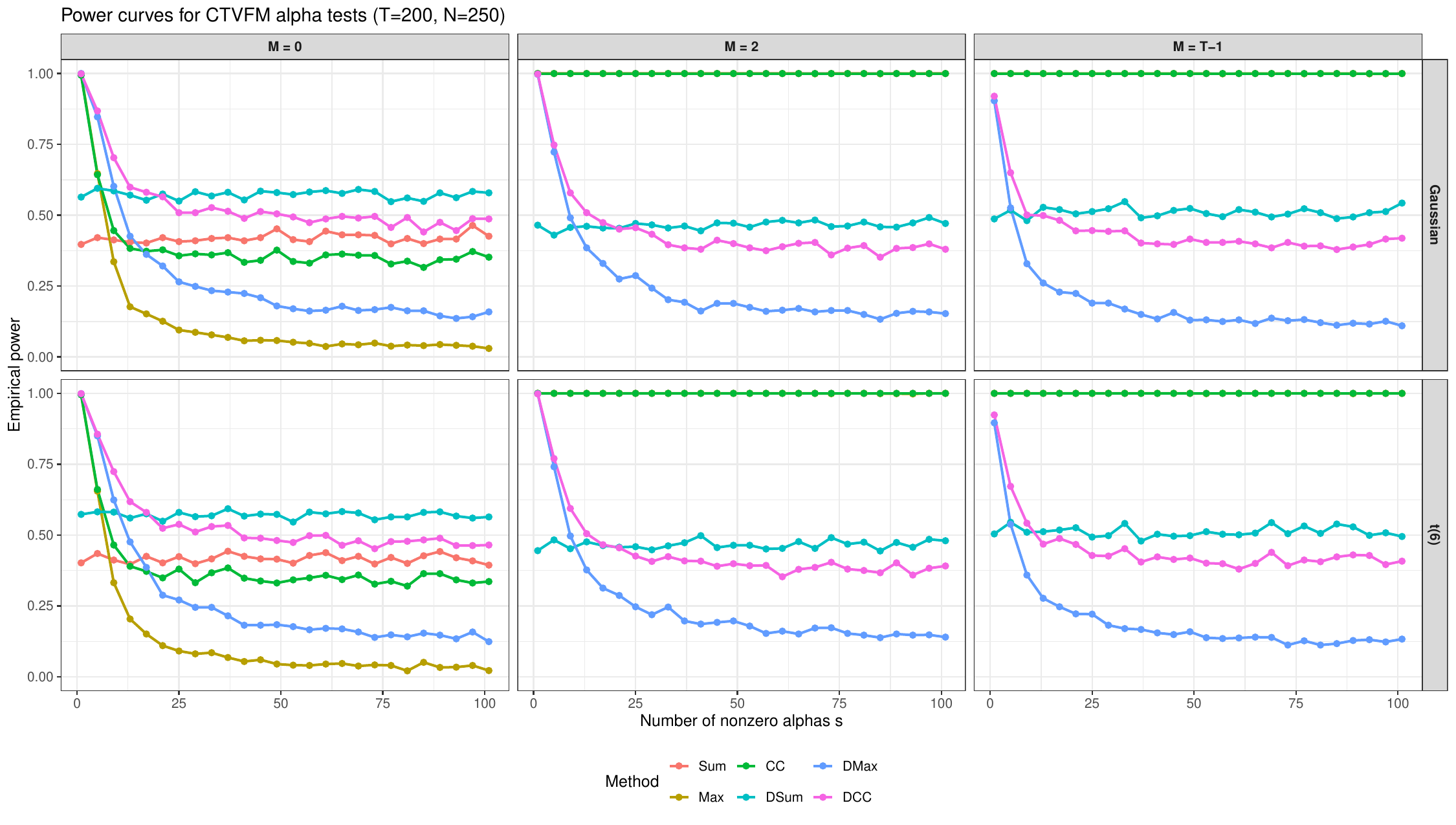}
\caption{Preliminary power curves for Example~2 under different sparse alternatives.}
\label{fig:power_preliminary_ex2}
\end{figure}

Figures~\ref{fig:power_preliminary_ex1} and \ref{fig:power_preliminary_ex2} deliver a coherent message across the two examples, the three dependence regimes, and the two innovation distributions. In the benchmark case $M=0$, where serial dependence is absent, the dependence-adjusted procedures remain fully competitive with the original methods. More specifically, the relative behavior of the adjusted procedures is well aligned with their theoretical roles: \textrm{DMAX} is the most effective in the very sparse regime, \textrm{DSUM} becomes increasingly advantageous as the alternative becomes denser, and \textrm{DCC} is particularly attractive because it remains highly stable in the sparse region while still preserving competitive performance more broadly. This pattern is visible in both examples and under both Gaussian and $t(6)$ errors, indicating that the proposed adjustment does not sacrifice power when dependence is weak, while still preserving the expected complementarity between max-type, sum-type, and combination-type procedures.

The advantage of the adjusted procedures becomes much more pronounced once temporal dependence is introduced. For both $M=2$ and $M=T-1$, the original methods display rejection frequencies that are often excessively large and nearly flat across different sparsity levels, especially for the combination procedure and, to a lesser extent, for the sum-based procedure. Such behavior is difficult to interpret as genuine discriminatory power, and is instead broadly consistent with the substantial size distortions documented in the size experiments. By contrast, the dependence-adjusted methods continue to exhibit economically and statistically meaningful power curves. In particular, \textrm{DMAX} retains its expected strength in detecting very sparse departures, whereas \textrm{DSUM} becomes the most reliable choice once the signal is spread over a larger fraction of coordinates. The adjusted combination procedure \textrm{DCC} is especially appealing because it remains highly robust in the sparse regime and continues to provide a favorable compromise between sensitivity and stability across a wide range of alternatives. Taken together, these findings suggest that explicitly accounting for serial dependence is indispensable for reliable inference in the conditional time varying factor model. From a practical perspective, \textrm{DSUM} appears to be the preferred choice when dense or moderately dense alternatives are of primary concern, \textrm{DMAX} is particularly useful under strong sparsity, and \textrm{DCC} offers a robust all-around procedure, especially when the degree of sparsity is unknown in advance.

\section{Real Data Application}
\label{sec:empirical_sp500}

We illustrate the practical relevance of dependence-robust inference using weekly returns on S\&P 500 constituent stocks. The stock-price data are downloaded from Yahoo Finance and converted into weekly returns. The factor data are the weekly Fama--French three factors, obtained from Kenneth French's online data library. After aligning the two data sources by calendar week, the common sample runs from January 14, 2005 to January 5, 2024, yielding $T=894$ weekly observations. Because the high-dimensional tests require a balanced panel, we restrict attention to the $N=393$ stocks with complete weekly returns throughout the full overlap period.

For each stock $i$ and week $t$, we estimate the conditional time-varying three-factor model
\begin{equation}
R_{it}-R_{ft}
= \alpha_i(t/T)
+ \beta_{iM}(t/T)\,\text{MKT}_t
+ \beta_{iS}(t/T)\,\text{SMB}_t
+ \beta_{iH}(t/T)\,\text{HML}_t
+ u_{it},
\label{eq:empirical_ctvfm}
\end{equation}
where $R_{it}$ is the stock return, $R_{ft}$ is the one-week Treasury bill rate, $\text{MKT}_t$ is the market excess return, $\text{SMB}_t$ is the size factor (small minus big), and $\text{HML}_t$ is the value factor (high minus low). The coefficient functions are allowed to vary smoothly over time and are approximated using the same cubic B-spline specification employed in the simulation analysis.

Before conducting alpha tests, it is useful to examine whether the residual process is close to white noise. Figure~\ref{fig:bp_histogram_sp500} plots the histogram of stock-level Box--Pierce $p$-values from the fitted time-varying three-factor model. The distribution is heavily tilted toward small $p$-values: the mean and median $p$-values are 0.252 and 0.122, respectively, and 35.1\% of the stocks reject the white-noise null at the 5\% level. This rejection rate is far above the nominal size and provides direct empirical evidence that residual serial dependence is pervasive in this panel. Consequently, inference procedures that ignore time-series dependence may be misleading in this application.

\begin{figure}[!htbp]
\centering
\includegraphics[width=0.72\textwidth]{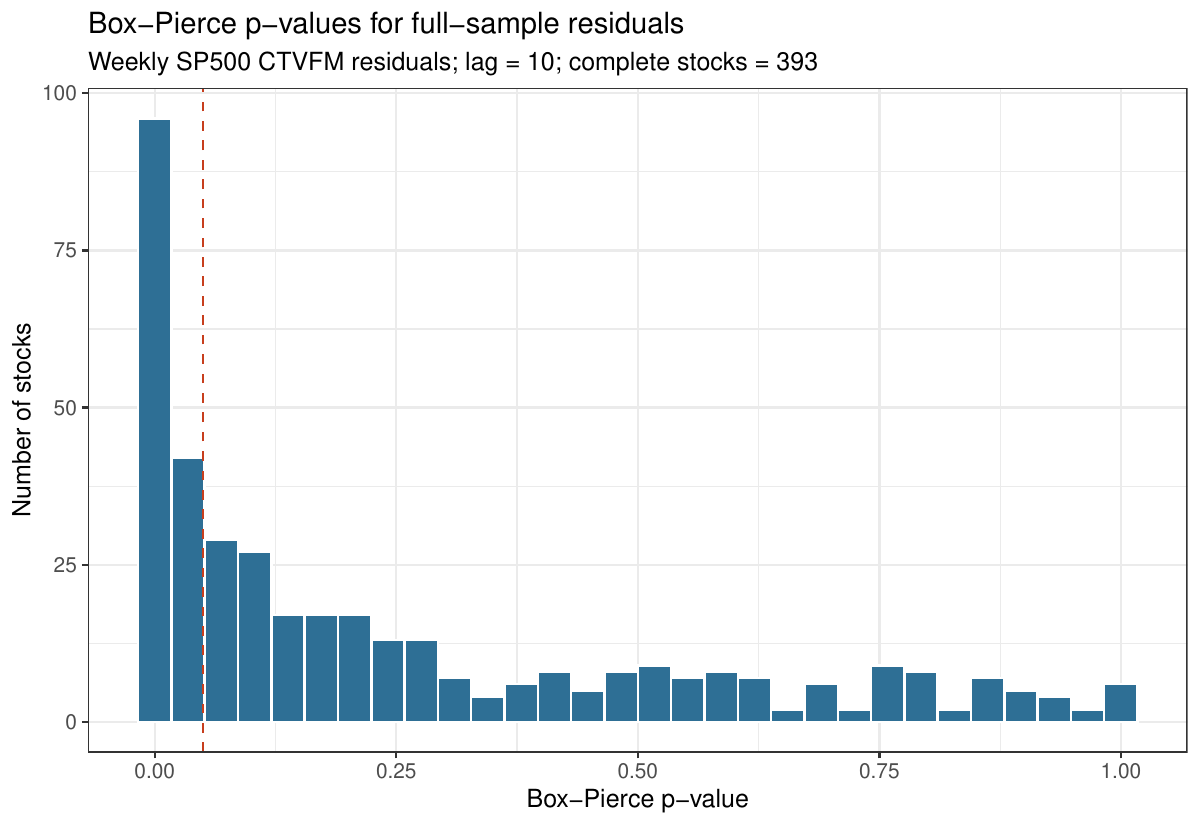}
\caption{Histogram of Box--Pierce $p$-values from the residuals of the time-varying Fama--French three-factor model. The vertical dashed line marks the 5\% level.}
\label{fig:bp_histogram_sp500}
\end{figure}

Table~\ref{tab:sp500_fullsample_pvalues} reports the full-sample $p$-values for the six testing procedures. The dependence-robust procedures are implemented using a moving block bootstrap with block length 18. The evidence is striking. Among the classical methods, \textrm{MAX} and \textrm{CC} strongly reject the joint null that all alphas are zero, with $p$-values below $10^{-6}$, whereas \textrm{SUM} does not reject. Once serial dependence is explicitly accounted for, however, the three dependence-robust procedures---\textrm{DSUM}, \textrm{DMAX}, and \textrm{DCC}---all produce comfortably insignificant $p$-values. In particular, \textrm{DMAX} yields a $p$-value of 0.920 and \textrm{DCC} yields a $p$-value of 0.850. The empirical message is therefore clear: if one ignores time-series dependence, one may obtain spurious evidence against the no-alpha null and overstate the degree of cross-sectional mispricing.

\begin{table}[!htbp]
\centering
\caption{Full-sample $p$-values for weekly S\&P 500 returns under the time-varying Fama--French three-factor model}
\label{tab:sp500_fullsample_pvalues}
\begin{tabular}{lc}
\hline
Method & $p$-value \\
\hline
\textrm{SUM}  & 0.8053 \\
\textrm{MAX}  & $2.13\times 10^{-7}$ \\
\textrm{CC}   & $4.26\times 10^{-7}$ \\
\textrm{DSUM} & 0.5163 \\
\textrm{DMAX} & 0.9198 \\
\textrm{DCC}  & 0.8503 \\
\hline
\end{tabular}
\vspace{0.4em}

\parbox{0.9\textwidth}{\small \textit{Notes}: The sample covers January 14, 2005 to January 5, 2024 and contains $T=894$ weekly observations and $N=393$ stocks with complete returns over the full sample. The dependence-robust procedures use a moving block bootstrap with block length 18.}
\end{table}

Taken together, these findings reinforce the main methodological point of the paper. In a large panel of weekly equity returns, residual dependence is empirically important rather than a minor second-order feature. As a result, procedures calibrated under independence can lead to qualitatively different conclusions from procedures that are robust to serial dependence. For this dataset, the dependence-robust tests suggest that the apparent evidence against zero pricing errors is largely driven by misspecified temporal dependence rather than by economically meaningful nonzero alphas.

\section{Discussion}
\label{sec:discussion}

This paper studies alpha testing in a high-dimensional conditional time-varying factor model with temporally dependent observations. Using a B-spline sieve approximation, we develop sum-type, max-type, and Cauchy combination tests that are designed to handle dense, sparse, and unknown alternatives, respectively. We establish the asymptotic theory for these procedures under temporal dependence and show through simulations and an empirical application that the proposed methods perform well in finite samples and are useful in practice.

Several directions deserve further study. One important extension is to allow for heavy-tailed innovations, which are common in financial returns and may require robustification of both the test statistics and the bootstrap calibration \citep{liu2023robust,zhao2023robust}. It would also be of interest to consider more general forms of dynamic structure, such as stronger serial dependence \citep{zhang2017gaussian}, structural breaks\citep{cho2015multiple}. An important practical issue is the choice of block length in the bootstrap procedure. As usual, this involves a trade-off between size and power: shorter blocks may not capture serial dependence adequately and can lead to size distortions, whereas longer blocks better preserve dependence but may yield a more variable bootstrap approximation and hence lower power in finite samples. Developing theoretically grounded and practically effective block-length selection rules for high-dimensional dependent factor models is therefore an interesting direction for future work.

\section{Appendix}
\label{sec:proofs}

\subsection{Auxiliary lemmas}

We collect the deterministic and probabilistic ingredients needed in the proofs.

\begin{lemma}[Spline approximation]
\label{lem:spline}
Under Assumption \ref{ass:smooth}, there exist coefficient vectors $\lambda_{i0}^0,\dots,\lambda_{id}^0$ such that
$$
\max_{1\le i\le N}\max_{1\le t\le T}|r_{it}|\le C L^{-r}\Bigl(1+\max_{1\le t\le T}\norm{\f_t}_\infty\Bigr).
$$
Hence under Assumption \ref{ass:factors},
$$
\max_{1\le i\le N}\max_{1\le t\le T}|r_{it}|\le C L^{-r}
$$
with probability approaching one.
\end{lemma}

\begin{proof}
This is the standard sieve approximation bound for H\"older-smooth varying-coefficient functions (See Lemma A.1 in \citet{ma2020testing} for details). The factor multiplier is uniformly bounded by Assumption \ref{ass:factors}. 
\end{proof}

\begin{lemma}
\label{lem:proj}
{ 
Recall 
\begin{align*}
\eta_t=\frac{1-\Z_t^{\T}(\Z^{\T}\Z/T)^{-1}(\Z^{\T}\bone_T/T)}{1-(\bone_T^{\T}\Z/T)(\Z^{\T}\Z/T)^{-1}(\Z^{\T}\bone_T/T)}
~~\text{and}~~
\wk\eta_t=\frac{1-\Z_t^{\T}\{\E(\Z^{\T}\Z/T)\}^{-1}\E(\Z^{\T}\bone_T/T)}{1-\E(\bone_T^{\T}\Z/T)\{\E(\Z^{\T}\Z/T)\}^{-1}\E(\Z^{\T}\bone_T/T)}.
\end{align*}
Under Assumption \ref{ass:smooth}- \ref{ass:factors}, we have with probability approaching one,
\begin{align*}
\max_{1\le t\le T}|\eta_t|\le C,
\qquad
\max_{1\le t\le T}|\wk\eta_t|\le C,
\qquad
\max_{1\le t\le T}|\eta_t-\wk\eta_t|\le CT^{-1/2}L^{1/2}\log T.
\end{align*}}
\end{lemma}

\begin{proof}
{ 
Since $\bM_Z$ is idempotent, its eigenvalues are either 0 or 1. Then $\a^{\T}\bM_Z\a\ge 0$ for any $\a\in\mR^{T}$ with $\norm{\a}=1$ and the equality holds when $\a$ is an eigenvector corresponding to $\lambda_{\min}(\bM_Z)=0$. Note that $\bone_T/\sqrt{T}$ is not an eigenvector, thus $\bone_T^{\T}\bM_Z\bone_T/T\ge c>0$. In addition, $\bone_T^{\T}\bM_Z\bone_T/T\le \lambda_{\max}(\bM_Z)\bone_T^{\T}\bone_T/T=1$. Hence, $\bone_T^{\T}\bM_Z\bone_T/T \asymp 1$.

By Bernstein's inequality in \citet{Bosq1996NonparametricSF} and the same proof for Lemma A.8 of \citet{Ma2011SplinebackfittedKS}, under $L^3T^{-1}=o(1)$,
\begin{equation*}
\norm{\Z^{\T}\bone_T/T-\E(\Z^{\T}\bone_T/T)}_\infty=O_{\text{a.s.}}(T^{-1/2}L^{-1/2}\log T).
\end{equation*} 
Further, under Assumption \ref{ass:factors}
\begin{align*}
\norm{\E(\Z^{\T}\bone_T/T)}_\infty
&\le\max_{1\le l\le L}T^{-1}  \sum_{t=1}^T|B_l(t/T) |\left(1+\sum_{j=1}^d\E|f_{jt}| \right)\\
&\le MT^{-1} \max_{1\le l\le L}\sum_{t\in\{t:|l(t)-l|\le q-1\}}|B_l(t/T)|=O(L^{-1}),
\end{align*} 
which leads to $\norm{\Z^{\T}\bone_T/T}_\infty=O_{\text{a.s.}}(L^{-1})$. By Lemma A.2 in \citet{ma2020testing} and the result in \citet{Demko1986SpectralBF}, we have with probability one,
\begin{equation*}
\norm{(\Z^{\T}\Z/T)^{-1}}_\infty\le CL
~~\text{and}~~
\norm{\{\E(\Z^{\T}\Z/T)\}^{-1}}_\infty\le CL.
\end{equation*}
Thus,
\begin{equation*}
\norm{(\Z^{\T}\Z/T)^{-1}-\{\E(\Z^{\T}\Z/T)\}^{-1}}_\infty=O_\text{a.s.}(L^2)\norm{\Z^{\T}\bone_T/T-\E(\Z^{\T}\bone_T/T)}_\infty=O_\text{a.s.}(T^{-1/2}L^{3/2}\log T).
\end{equation*}

Using the fact that $\sum_{l=1}^LB_l(t/T)$ is bounded, we have with probability approaching one,
\begin{align*}
|\eta_t|&=\left|\frac{1-\Z_t^{\T}(\Z^{\T}\Z)^{-1}\Z^{\T}\bone_T}{\bone_T^{\T}\bM_Z\bone_T/T}\right|\le C(1+\norm{\Z_t}_1\cdot ||(\Z^{\T}\Z/T)^{-1}||_\infty\cdot ||\Z^{\T}\bone_T/T||_\infty)\le C(1+\norm{\bm{f}_t}_\infty),
\end{align*}
\begin{align*}
\left|\eta_t-\frac{1-\Z_t^{\T}\{\E(\Z^{\T}\Z/T)\}^{-1}\E(\Z^{\T}\bone_T/T)}{\bone_T^{\T}\bM_Z\bone_T/T}\right|\le& C\norm{\Z_t}_1\cdot \norm{(\Z^{\T}\Z/T)^{-1}-\{\E(\Z^{\T}\Z/T)\}^{-1}}_\infty\cdot \norm{\Z^{\T}\bone_T/T}_\infty\\
&+C\norm{\Z_t}_1\cdot \norm{\{\E(\Z^{\T}\Z/T)\}^{-1}}_\infty\cdot \norm{\Z^{\T}\bone_T/T-\E(\Z^{\T}\bone_T/T)}_\infty\\
\le&CT^{-1/2}L^{1/2}\log T(1+\norm{\bm{f}_t}_\infty),
\end{align*}
and 
\begin{align*}
&\left|\wk\eta_t-\frac{1-\Z_t^{\T}\{\E(\Z^{\T}\Z/T)\}^{-1}\E(\Z^{\T}\bone_T/T)}{\bone_T^{\T}\bM_Z\bone_T/T}\right|\\
\le&|\wk\eta_t|\cdot\left|(\bone_T^{\T}\Z/T)(\Z^{\T}\Z/T)^{-1}(\Z^{\T}\bone_T/T)-\E(\bone_T^{\T}\Z/T)\{\E(\Z^{\T}\Z/T)\}^{-1}\E(\Z^{\T}\bone_T/T)\right|\\
\le&CT^{-1/2}L^{1/2}\log T(1+\norm{\bm{f}_t}_\infty).
\end{align*}
Then desired results follow from the boundedness of factor.}
\end{proof}

\begin{lemma}
\label{lem:trunc}
{ 
Recall
\begin{align*}
\X_t=\e_t\eta_t,\qquad
\wk \X_t=\e_t\wk \eta_t,\qquad
\wh{\X}_t=\wh{\be}_t\eta_t,
\qquad
\bar{\wh{\X}}=T^{-1}\sum_{t=1}^T \wh{\X}_t,
\qquad
\tilde{\X}_t=\wh{\X}_t-\bar{\wh{\X}}.
\end{align*}
Under Assumption \ref{ass:smooth}- \ref{ass:errors} and $H_0$, we have
\begin{align*}
    &\left|\wb\X_T^{\T}\wb \X_T-\wb{\wk \X}_T^{\T}\wb{\wk \X}_T\right|=\Op\left\{T^{-3/2}L^{1/2}(\log T)\tr(\bmS)\right\},\\
    &\max_{0\le h\le T}\left|T^{-1}\sum_{t=1}^T\left(\tilde{\X}_t^{\T}\tilde{\X}_{t+h}-\X_t^{\T}\X_{t+h}\right)\right|=\Op\left\{T^{-1/2}N^{1/2}L^{-r}\tr^{1/4}(\bmS^2)+T^{-1}L\tr(\bmS)\right\},\\
    &\max_{0\le h\le T}\left\|T^{-1}\sum_{t=1}^T\left(\tilde{\X}_{t+h}\tilde{\X}_t^{\T}-\X_{t+h}\X_t^{\T}\right)\right\|_F=\Op\{T^{-1/2}N^{1/2}L^{-r+1}\tr^{1/2}(\bmS)+T^{-1}L\tr(\bmS)\}.
\end{align*}
}
\end{lemma}

\begin{proof}
{ 
Observing that 
\begin{align*}
    \wb\X_T^{\T}\wb \X_T-\wb{\wk \X}_T^{\T}\wb{\wk \X}_T=2T^{-2}\sum_{1\le t,s\le T}\X_t^{\T}(\X_s-\wk \X_s)+T^{-2}\sum_{1\le t,s\le T}(\X_t-\wk \X_t)^{\T}(\X_s-\wk \X_s),
\end{align*}
in which the interaction term is the main term. Since $\{\f_t\}$ and $\{\e_t\}$ are independent,
\begin{align*}
    \E\left\{\left|T^{-2}\sum_{1\le t,s\le T}\X_t^{\T}(\X_s-\wk \X_s)\right|^2\right\}=T^{-4}\sum_{1\le t_1,t_2,t_3,t_4\le T}\E(\e_{t_1}^{\T}\e_{t_2}\e_{t_3}^{\T}\e_{t_4})\E\{\eta_{t_1}(\eta_{t_2}-\wk\eta_{t_2})\eta_{t_3}(\eta_{t_4}-\wk\eta_{t_4})\}.
\end{align*}
Using H\"{o}lder inequality and the results in Lemma~\ref{lem:proj}, 
\begin{align*}
    \max_{t_1,t_2,t_3,t_4}|\E\{\eta_{t_1}(\eta_{t_2}-\wk\eta_{t_2})\eta_{t_3}(\eta_{t_4}-\wk\eta_{t_4})\}|=O(T^{-1}L\log^2 T).
\end{align*}
Write $\e_t=\bmS^{1/2}\u_t$, where $\u_t=(u_{1t},\dots,u_{Nt})^{\T}$ with $u_{it}=\sum_{k=0}^\infty b_k\z_{i,t-k}$. Then for each $t$, $\{u_{it}\}_{i\ge 1}$ are independent with $\E(u_{it})=0$, and $\E(u_{it}u_{i,t+h})=a_h$ with
\begin{align}\label{eq:a_h}
a_h:=\sum_{k=0}^\infty b_kb_{k+h}.
\end{align}
Then,
\begin{align}\label{eq:fourmoment_e}
&\E(\e_{t_1}^{\T}\e_{t_2}\e_{t_3}^{\T}\e_{t_4})\nonumber\\
=&\sum_{1\le i,j,k,l\le N}\Sigma_{ij}\Sigma_{kl}\E(u_{it_1}u_{jt_2}u_{kt_3}u_{lt_4})\nonumber\\
=&\sum_{i\ne j}\Sigma_{ii}\Sigma_{jj}a_{t_1-t_2}a_{t_3-t_4}+\sum_{i\ne j}\Sigma_{ij}^2(a_{t_1-t_3}a_{t_2-t_4}+a_{t_1-t_4}a_{t_2-t_3})+\sum_{i}\Sigma_{ii}^2\E(u_{it_1}u_{it_2}u_{it_3}u_{it_4})\nonumber\\
\approx&\tr^2(\bmS)a_{t_1-t_2}a_{t_3-t_4}+\tr(\bmS^2)(a_{t_1-t_3}a_{t_2-t_4}+a_{t_1-t_4}a_{t_2-t_3})+\sum_{i}\Sigma_{ii}^2(a_{t_1-t_2}a_{t_3-t_4}+a_{t_1-t_3}a_{t_2-t_4}+a_{t_1-t_4}a_{t_2-t_3}).
\end{align}
Under the assumption that $\sum_{k=0}^\infty |b_k|<\infty$ and $b_k=o(k^{-5-\eta_b})$, we have
\begin{align}\label{eq:boundba}
\sum_{k=n}^\infty |b_k|=o(n^{-4-\eta_b}),\qquad 
T^{-1}\sum_{1\le t_1,t_2\le T}|a_{t_1-t_2}|< \infty.
\end{align}  
Hence,
\begin{align*}
T^{-4}\sum_{1\le t_1,t_2,t_3,t_4\le T}|\E(\e_{t_1}^{\T}\e_{t_2}\e_{t_3}^{\T}\e_{t_4})|\le CT^{-2}\tr^2(\bmS).   
\end{align*}
Accordingly,
\begin{align*}
\E\left\{\left|T^{-2}\sum_{1\le t,s\le T}\X_t^{\T}(\X_s-\wk \X_s)\right|^2\right\}=O\{T^{-3}L(\log^2 T)\tr^2(\bmS)\},
\end{align*}
which leads to 
\begin{align*}
    \left|\wb\X_T^{\T}\wb \X_T-\wb{\wk \X}_T^{\T}\wb{\wk \X}_T\right|=\Op\left\{T^{-3/2}L^{1/2}(\log T)\tr(\bmS)\right\}.
\end{align*}

Next, we analyze $T^{-1}\sum_{t=1}^T\left(\tilde{\X}_t^{\T}\tilde{\X}_{t+h}-\X_t^{\T}\X_{t+h}\right)$. Under $H_0$, $\wh{\e}_{i\cdot}=\bM_Z(\e_{i\cdot}+\r_{i\cdot})=(\I_T-\bP_Z)(\e_{i\cdot}+\r_{i\cdot})$, then
\begin{align}\label{eq:decome_it}
&\wh e_{it}=e_{it}+r_{it}+\Z_t^{\T}(\Z^{\T}\Z)^{-1}\Z^{\T}\r_{i\cdot}+\Z_t^{\T}(\Z^{\T}\Z)^{-1}\Z^{\T}\e_{i\cdot},~~\text{and}\nonumber\\
&\bar{\wh{X}}_i=\kappa_T^{-1}\h^{\T}\wh\e_{i\cdot}=\kappa_T^{-1}\h^{\T}(\e_{i\cdot}+\r_{i\cdot})=T^{-1}\bmeta^{\T}(\e_{i\cdot}+\r_{i\cdot}),
\end{align}
where the second equation uses the fact that $\h^{\T}\bM_Z=\h^{\T}$. Thus, 
\begin{align*}
T^{-1}\sum_{t=1}^T&\left(\tilde{\X}_t-\X_t\right)^{\T}\X_{t+h}\\
=T^{-1}\sum_{t=1}^T&\sum_{i=1}^N\left\{r_{it}e_{i,t+h}\eta_t\eta_{t+h}
+\Z_t^{\T}(\Z^{\T}\Z)^{-1}\Z^{\T}\r_{i\cdot}e_{i,t+h}\eta_t\eta_{t+h}
+\Z_t^{\T}(\Z^{\T}\Z)^{-1}\Z^{\T}\e_{i\cdot}e_{i,t+h}\eta_t\eta_{t+h}\right.\\
&\left.\qquad+T^{-1}\bmeta^{\T}\e_{i\cdot}e_{i,t+h}\eta_{t+h}+T^{-1}\bmeta^{\T}\r_{i\cdot}e_{i,t+h}\eta_{t+h}\right\}\\
\qquad=:I_{h,1}+&I_{h,2}+I_{h,3}+I_{h,4}+I_{h,5}.
\end{align*}
Since $\{\f_t\}_{t\ge 1}$ and $\{\e_t\}_{t\ge 1}$ are independent,
\begin{align*}
\max_h\E(I_{h,1}^2)
=&\max_hT^{-2}\sum_{1\le t,s\le T}\E(\r_t^{\T}\e_{t+h}\r_s^{\T}\e_{s+h}\eta_t\eta_{t+h}\eta_s\eta_{s+h})\\
=&\max_hT^{-2}\sum_{1\le t,s\le T}\tr\left\{\E(\e_{t+h}\e_{s+h}^{\T})\E(\r_s\r_t^{\T}\eta_t\eta_{t+h}\eta_s\eta_{s+h})\right\}\\
\le&\max_hT^{-2}\sum_{1\le t,s\le T}\norm{\E(\e_{t+h}\e_{s+h}^{\T})}_F\cdot\norm{\E(\r_s\r_t^{\T}\eta_t\eta_{t+h}\eta_s\eta_{s+h})}_F\\
\le&CT^{-2}\sum_{1\le t,s\le T}|a_{t-s}|\cdot\norm{\bmS}_F\cdot NL^{-2r}\\
\le&CT^{-1}NL^{-2r}\tr^{1/2}(\bmS^2),
\end{align*}
where the last second step comes from the fact that $\E(\e_t\e_{t+h}^{\T})=a_h\bmS$, the results in Lemma~\ref{lem:spline} and \ref{lem:proj}, and \eqref{eq:boundba}.
\begin{align*}
&\max_h\E(I_{h,2}^2)\\
=&\max_hT^{-2}\sum_{1\le t,s\le T}\sum_{1\le i,j\le N}\E(e_{i,t+h}e_{j,s+h})\E\left\{\Z_t^{\T}(\Z^{\T}\Z)^{-1}\Z^{\T}\r_{i\cdot}\Z_t^{\T}(\Z^{\T}\Z)^{-1}\Z^{\T}\r_{j\cdot}\eta_t\eta_{t+h}\eta_s\eta_{s+h}\right\}\\
\le&T^{-2}\sum_{1\le t,s\le T}\sum_{1\le i,j\le N}|a_{t-s}\Sigma_{ij}|\cdot\max_{i,t,h}\left|\E\left\{\Z_t^{\T}(\Z^{\T}\Z)^{-1}\Z^{\T}\r_{i\cdot}\Z_t^{\T}(\Z^{\T}\Z)^{-1}\Z^{\T}\r_{j\cdot}\eta_t\eta_{t+h}\eta_s\eta_{s+h}\right\}\right|.
\end{align*}
According to the proof of Lemma~\ref{lem:trunc}, we have with probability approaching one
\begin{align}\label{eq:PZrsize}
\max_{i,t}\left|\Z_t^{\T}(\Z^{\T}\Z)^{-1}\Z^{\T}\r_{i\cdot}\right|
\le&\max_t\norm{\Z_t}_1\cdot\norm{(\Z^{\T}\Z/T)^{-1}}_\infty\cdot\max_i\norm{\Z^{\T}\r_{i\cdot}/T}_\infty\nonumber\\
\le&CL\max_{1\le k\le(d+1)L}T^{-1}\sum_{t=1}^T|Z_{tk}|\cdot\max_{i,t}|r_{it}|\nonumber\\
\le&CL^{1-r}\max_{1\le l\le L}T^{-1}\sum_{t=1}^T|B_l(t/T)|\left(1+\max_{1\le j\le d}|f_{jt}|\right)\le CL^{-r},
\end{align}
which, combining with the fact that $\norm{\bmS}_1^{{\rm entry}}\le N\norm{\bmS}_F$, leads to
\begin{align*}
\max_h\E(I_{h,2}^2)\le CT^{-1}N\tr^{1/2}(\bmS^2)L^{-2r}.
\end{align*}
Note that $I_{h,3}=T^{-1}\sum_{1\le t,s\le T}\e_s^{\T}\e_{t+h}[\bP_Z]_{ts}\eta_t\eta_{t+h}$, by \eqref{eq:fourmoment_e},  
\begin{align*}
\max_h\E(I_{h,3}^2)
=T^{-2}\sum_{1\le t,s,m,n\le T}\E(\e_s^{\T}\e_{t+h}\e_n^{\T}\e_{m+h})\E\left([\bP_Z]_{ts}[\bP_Z]_{mn}\eta_t\eta_{t+h}\eta_m\eta_{m+h}\right)\le CT^{-2}L^2\tr^2(\bmS),
\end{align*}
where the final step uses the fact that
\begin{align}\label{eq:PZsize}
\max_{t,s}|[\bP_Z]_{ts}|\le T^{-1}\max_t\norm{\Z_t}_1\cdot\norm{(\Z^{\T}\Z/T)^{-1}}_\infty\cdot\max_s\norm{\Z_s}_\infty=\Op(T^{-1}L).
\end{align}
For $I_{h,4}=T^{-2}\sum_{1\le t,s\le T}\e_t^{\T}\e_{t+h}\eta_s\eta_{t+h}$, by \eqref{eq:fourmoment_e}, 
\begin{align*}
\max_h\E(I_{h,4}^2)
=&\max_hT^{-4}\sum_{1\le t,s,m,n\le T}\E(\e_s^{\T}\e_{t+h}\e_n^{\T}\e_{m+h})\E(\eta_s\eta_{t+h}\eta_n\eta_{m+h})\le CT^{-2}\tr^2(\bmS).
\end{align*}
For $I_{h,5}=T^{-2}\sum_{1\le t,s\le T}\r_s^{\T}\e_{t+h}\eta_s\eta_{t+h}$, similarly as $\max_h\E(I_{h,1}^2)$, we have
\begin{align*}
\max_h\E(I_{h,5}^2)
=&\max_hT^{-4}\sum_{1\le t,s,m,n\le T}\tr\{\E(\e_{t+n}\e_{m+h}^{\T})\E(\r_n\r_s^{\T}\eta_{t+h}\eta_s\eta_{m+h}\eta_n)\}\le CT^{-1}NL^{-2r}\tr^{1/2}(\bmS^2).
\end{align*}

Combining all the bounds of the second moment, we have
\begin{align*}
\max_h\left|T^{-1}\sum_{t=1}^T\left(\tilde{\X}_t-\X_t\right)^{\T}\X_{t+h}\right|=\Op\left\{T^{-1/2}N^{1/2}L^{-r}\tr^{1/4}(\bmS^2)+T^{-1}L\tr(\bmS)\right\}, 
\end{align*}
which is the leading term in $\max_h\left|T^{-1}\sum_{t=1}^T\left(\tilde{\X}_t^{\T}\tilde{\X}_{t+h}-\X_t^{\T}\X_{t+h}\right)\right|$.

Finally, we analyze $\left\|T^{-1}\sum_{t=1}^T\left(\tilde{\X}_{t+h}\tilde{\X}_t^{\T}-\X_{t+h}\X_t^{\T}\right)\right\|_F$. By \eqref{eq:decome_it},
\begin{align*}
&\left\|T^{-1}\sum_{t=1}^T\X_{t+h}\left(\tilde{\X}_t-\X_t\right)^{\T}\right\|_F^2\\
=&\sum_{1\le i,j\le N}\left|T^{-1}\sum_{t=1}^TX_{j,t+h}(\tilde{X}_{it}-X_{it})\right|^2\\
=&\sum_{1\le i,j\le N}\left|T^{-1}\sum_{t=1}^T\left\{r_{it}+\Z_t^{\T}(\Z^{\T}\Z)^{-1}\Z^{\T}\r_{i\cdot}+\Z_t^{\T}(\Z^{\T}\Z)^{-1}\Z^{\T}\e_{i\cdot}-T^{-1}\bmeta^{\T}(\e_{i\cdot}+\r_{i\cdot})\right\}e_{j,t+h}\eta_t\eta_{t+h}\right|^2\\
\le&5\sum_{1\le i,j\le N}(I_{1,ijh}^2+I_{2,ijh}^2+I_{3,ijh}^2+I_{4,ijh}^2+I_{5,ijh}^2),
\end{align*}
where 
\begin{align*}
I_{1,ijh}=&T^{-1}\sum_{t=1}^Tr_{it}e_{j,t+h}\eta_{t}\eta_{t+h},\\
I_{2,ijh}=&T^{-1}\sum_{t=1}^T\Z_t^{\T}(\Z^{\T}\Z)^{-1}\Z^{\T}\r_{i\cdot}e_{j,t+h}\eta_{t}\eta_{t+h}=T^{-1}\sum_{1\le t,s\le T}e_{j,t+h}r_{is}[\bP_Z]_{ts}\eta_{t}\eta_{t+h},\\ 
I_{3,ijh}=&T^{-1}\sum_{t=1}^T\Z_t^{\T}(\Z^{\T}\Z)^{-1}\Z^{\T}\e_{i\cdot}e_{j,t+h}\eta_{t}\eta_{t+h}=T^{-1}\sum_{1\le t,s\le T}e_{j,t+h}e_{is}[\bP_Z]_{ts}\eta_{t}\eta_{t+h},\\ 
I_{4,ijh}=&T^{-1}\sum_{t=1}^TT^{-1}\bmeta^{\T}\e_{i\cdot}e_{j,t+h}\eta_{t}\eta_{t+h},\\ 
I_{5,ijh}=&T^{-1}\sum_{t=1}^TT^{-1}\bmeta^{\T}\r_{i\cdot}e_{j,t+h}\eta_{t}\eta_{t+h}.
\end{align*}
By Lemma~\ref{lem:spline} and \ref{lem:proj}, the fact that $\E(e_{j,t+h}e_{j,s+h})=\Sigma_{jj}a_{t-s}$ and \eqref{eq:boundba}, we have
\begin{align*}
\max_h\E(I_{1,ijh}^2)
=&\max_hT^{-2}\sum_{1\le t,s\le T}\E(e_{j,t+h}e_{j,s+h})\E(r_{it}r_{is}\eta_{t}\eta_{t+h}\eta_{s}\eta_{s+h})\le CT^{-1}L^{-2r}\Sigma_{jj},
\end{align*}
which leads to $\max_h\sum_{1\le i,j\le N}I_{1,ijh}^2=\Op\{T^{-1}NL^{-2r}\tr(\bmS)\}$. Further, by \eqref{eq:PZsize},
\begin{align*}
\max_h\E(I_{2,ijh}^2)
=&\max_hT^{-2}\sum_{1\le t,s,m,n\le T} \E(e_{j,t+h}e_{j,m+h})\E (r_{is}r_{in}[\bP_Z]_{ts}[\bP_Z]_{mn}\eta_{t}\eta_{t+h}\eta_{m}\eta_{m+h})\le CT^{-1}L^{-2r+2}\Sigma_{jj},
\end{align*}
which leads to $\max_h\sum_{1\le i,j\le N}I_{2,ijh}^2=\Op\{T^{-1}NL^{-2r+2}\tr(\bmS)\}$. For
$I_{3,ijh}$, noting that
\begin{align*}
\sum_{1\le i,j\le N}I_{3,ijh}^2
=T^{-2}\sum_{1\le t,s,m,n\le T}\e_s^{\T}\e_n\e_{t+h}^{\T}\e_{m+h}[\bP_Z]_{ts}[\bP_Z]_{mn}\eta_{t}\eta_{t+h}\eta_{m}\eta_{m+h},     
\end{align*}
by \eqref{eq:fourmoment_e}, \eqref{eq:boundba} and \eqref{eq:PZsize}, we have
\begin{align*}
\max_h\E\left(\left|\sum_{1\le i,j\le N}I_{3,ijh}^2\right|\right)
=&\max_hT^{-2}\sum_{1\le t,s,m,n\le T}\E(\e_s^{\T}\e_n\e_{t+h}^{\T}\e_{m+h})\E([\bP_Z]_{ts}[\bP_Z]_{mn}\eta_{t}\eta_{t+h}\eta_{m}\eta_{m+h})\\
\le&CT^{-2}L^2\tr^2(\bmS),
\end{align*}
which leads to $\max_h\sum_{1\le i,j\le N}I_{3,ijh}^2=\Op\{T^{-2}L^2\tr^2(\bmS)\}$. Similarly, for $I_{4,ijh}$, we have
\begin{align*}
\max_h\E\left(\left|\sum_{1\le i,j\le N}I_{4,ijh}^2\right|\right)
=&\max_hT^{-4}\sum_{1\le t,s,m,n\le T}\E(\e_m^{\T}\e_n\e_{t+h}^{\T}\e_{s+h})\E(\eta_n\eta_n\eta_t\eta_{t+h}\eta_s\eta_{s+h})\\
\le&CT^{-2}\tr^2(\bmS),
\end{align*}
which leads to $\max_h\sum_{1\le i,j\le N}I_{4,ijh}^2=\Op\{T^{-2}\tr^2(\bmS)\}$. For $T_{5,ijh}$, noting that
\begin{align*}
\sum_{1\le i,j\le N}I_{5,ijh}^2
=T^{-4}\sum_{1\le t,s,m,n \le T}\e_{t+h}^{\T}\e_{s+h}\left(\sum_{i=1}^Nr_{im}r_{in}\right) \eta_{t}\eta_{t+h}\eta_{s}\eta_{s+h}\eta_m\eta_n,     
\end{align*}
we have
\begin{align*}
\max_h\E\left(\left|\sum_{1\le i,j\le N}I_{5,ijh}^2\right|\right)
=&\max_hT^{-4}\sum_{1\le t,s,m,n \le T}\E(\e_{t+h}^{\T}\e_{s+h})\E\left\{\left(\sum_{i=1}^Nr_{im}r_{in}\right) \eta_{t}\eta_{t+h}\eta_{s}\eta_{s+h}\eta_m\eta_n\right\}\\
\le&CT^{-2}NL^{-2r}\cdot\max_h\sum_{1\le t,s\le T}|\E(\e_{t+h}^{\T}\e_{s+h})| \\
\le&CT^{-2}NL^{-2r}\cdot\sum_{i=1}^N \max_h\sum_{1\le t,s\le T}|\E(e_{i,t+h}e_{i,s+h})| \\
\le&CT^{-1}NL^{-2r}\tr(\bmS).
\end{align*}
Thus, $\max_h\sum_{1\le i,j\le N}I_{5,ijh}^2=\Op\{T^{-1}NL^{-2r}\tr(\bmS)\}$. In summary,
\begin{align*}
&\max_h\left\|T^{-1}\sum_{t=1}^T\X_{t+h}\left(\tilde{\X}_t-\X_t\right)^{\T}\right\|_F=\Op\{T^{-1/2}N^{1/2}L^{-r+1}\tr^{1/2}(\bmS)+T^{-1}L\tr(\bmS)\},
\end{align*}
which is the leading term in $\max_h\left\|T^{-1}\sum_{t=1}^T\left(\tilde{\X}_{t+h}\tilde{\X}_t^{\T}-\X_{t+h}\X_t^{\T}\right)\right\|_F$
}

\end{proof}

% \begin{lemma}[Long-run variance order]
% \label{lem:sigmaorder}
% Under Assumptions \ref{ass:errors} and \ref{ass:sigma},
% $$
% \sigma_T^2\asymp T^{-2}\tr(\bO_T^2),
% \qquad
% T^{-2}\tr(\bO_T^2)\ge c\frac{\tr(\bmS^2)}{T^2}
% $$
% for some constant $c>0$.
% \end{lemma}

% \begin{proof}
% Since $\bO_T$ is a weighted average of the autocovariance sequence of the linear process and $\sum_k b_k=s\neq 0$, its spectral density at frequency zero is nondegenerate. The trace bound follows from Assumption \ref{ass:sigma} and the same argument as in the sum-type theory of \citet{ma2024dependent}. 
% \end{proof}

\subsection{Proof of Theorem \ref{thm:sumclt}}
Under $H_0$,
\begin{align}
T_{\mathrm{DSUM}}
=(\wb{\X}_T+\b_T)^{\T}(\wb{\X}_T+\b_T)=\wb{\X}_T^{\T}\wb{\X}_T+2\b_T^{\T}\wb{\X}_T+\b_T^{\T} \b_T.
\label{eq:sumdecomp_pf}
\end{align}
According to the proof of Lemma~\ref{lem:proj}, we have $\kappa_T=\bone_T^{\T}\bM_Z\bone_T \asymp T$ and $\max_{1\le t\le T}|h_t|=\Op(1)$, which combining with the result of Lemma~\ref{lem:spline} that $\max_{i,t}|r_{it}|=\Op(L^{-r})$, imply $\norm{\b_T}=\Op(N^{1/2}L^{-r})$. Thus, with probability approaching one 
\begin{align}\label{eq:smallzero}
|\b_T^{\T} \b_T|\le CNL^{-2r}
~~\text{and}~~
|\b_T^{\T}\wb{\X}_T|\le CN^{1/2}L^{-r}\left(\wb{\X}_T^{\T}\wb{\X}_T\right)^{1/2}.
\end{align}
{ 
By Lemma~\ref{lem:trunc}, we have
\begin{align}\label{eq:errorcheck}
    \left|\wb\X_T^{\T}\wb \X_T-\wb{\wk \X}_T^{\T}\wb{\wk \X}_T\right|=\Op\left\{T^{-3/2}L^{1/2}(\log T)\tr(\bmS)\right\},
\end{align}
where 
\begin{align*}
\wk \X_t=\e_t\wk\eta_t,\qquad \wk\eta_t=\frac{1-\Z_t^{\T}\{\E(\Z^{\T}\Z/T)\}^{-1}\E(\Z^{\T}\bone_T/T)}{1-\E(\bone_T^{\T}\Z/T)\{\E(\Z^{\T}\Z/T)\}^{-1}\E(\Z^{\T}\bone_T/T)}.
\end{align*}
Under Assumption~\ref{ass:factors}, $\{\f_t\}_{t\ge 1}$ is stationary,  alpha mixing and independent of $\{\e_t\}_{t\ge 1}$, so $\{\wk\eta_t\}_{t\ge 1}$ is as well and $\alpha_\eta(k)=\alpha_f(k)$. By the mixing property (for example \cite{PT85mixing}),  $\e_t=\bmS^{1/2}\sum_{k=0}^\infty b_k\z_{t-k}$ is alpha mixing with $\alpha_e(k)=o(k^{-4-\eta_b})$. Next, we use the blocking method to prove the asymptotic normality of $\wb{\wk \X}_T^{\T}\wb{\wk \X}_T$. For any $T$, choose $\zeta\in(0,1)$ and $c>0$ such that $w_T=cT^\zeta>M\to\infty$ and $T=w_Tq_T+r_T$, where $0\le r_T<w_T$. For $1\le t,s\le q_T$, define
\begin{align*}
    &B_{ts}:=\sum_{k=(t-1)w_T+1}^{tw_T-M}\sum_{l=(s-1)w_T+1}^{sw_T-M}T^{-2}\wk\X_k^{\T}\wk\X_l,\\
    &D_{ts}:=\sum_{k=(t-1)w_T+1}^{tw_T}\sum_{l=(s-1)w_T+1}^{sw_T}T^{-2}\wk\X_k^{\T}\wk\X_l-B_{ts},\\
    &F:=\sum_{(k,l)\in\{1,\dots,T\}^2-\{1,\dots,q_Tw_T\}^2 }T^{-2}\wk\X_k^{\T}\wk\X_l.
\end{align*}
Then,
\begin{align}\label{eq:decompose}
    \wb{\wk \X}_T^{\T}\wb{\wk \X}_T=\sum_{1\le t\ne s\le q_T}B_{ts}+\sum_{t=1}^{q_T}B_{tt}+\sum_{1\le t,s\le q_T}D_{ts}+F.
\end{align}

For the first term, by Berbee's Lemma (see Lemma 1.1 in \cite{Bosq1996NonparametricSF}),  there exists a sequence of mutually independent random vectors $\{\wk\X_k^{I},k=(t-1)w_T+1,\dots,tw_T-M\}_{t=1}^{q_T}$ such that $\{\wk\X_k^{I},k=(t-1)w_T+1,\dots,tw_T-M\}_{t=1}^{q_T}$ and $\{\wk\X_k,k=(t-1)w_T+1,\dots,tw_T-M\}_{t=1}^{q_T}$ are identically distributed and 
\begin{align*}
\Pr\left(\{\wk\X_k^{I},k=(t-1)w_T+1,\dots,tw_T-M\}\ne \{\wk\X_k,k=(t-1)w_T+1,\dots,tw_T-M\}\right)\le C\alpha_e(M)\vee\alpha_f(M).
\end{align*}
Define $B_{ts}^I$ as $B_{ts}$ by replacing $\X_k$ with $\X_k^I$. Then for any $\epsilon>0$,
\begin{align*}
    \Pr\left(\left|\sum_{1\le t\ne s\le q_T}B_{ts}-\sum_{1\le t\ne s\le q_T}B_{ts}^I\right|>\epsilon\sqrt{\var\left(\sum_{1\le t\ne s\le q_T}B_{ts}^I\right)} \right)
    \le Cq_T^2(w_T-M)^2\{\alpha_e(M)\vee\alpha_f(M)\}\to 0,
\end{align*}
by properly setting $M\to\infty$. For $1\le l\le Q_t$, define 
\begin{align*}
M_l=\frac{(w_T-M)^2}{T^2}\sum_{t=2}^l\sum_{s=1}^{t-1}\Y_t^{\T}\Y_s
~~\text{with}~~
\Y_t=\frac{1}{w_T-M}\sum_{k=(t-1)w_T+1}^{tw_T-M}\wk\X_k^I,
\end{align*}
and $\mathcal{F}_l=\sigma(\Y_1,\dots,\Y_l)$. Then,
\begin{align*}
\sum_{1\le t\ne s\le q_T}B_{ts}^I=2M_{q_T}.  
\end{align*}
Since $\Y_1,\dots,\Y_{q_T}$ are independent with zero mean, $\{M_l,\mathcal{F}_l\}_{2\le l\le q_T}$ is a zero mean and square integrable martingale sequence.

Define
\begin{align}\label{eq:omegacheck}
\wk\bO_{\ell}=\sum_{0\le |h|\le \ell}\left(1-\frac{|h|}{\ell}\right)\wk\bG_h,
\qquad
\wk\bG_h=\E(\wk\X_t\wk\X_{t+h}^{\T}).
\end{align}
Note that $\wk\bG_h=\Gamma_{f,h}a_h\bmS$ with $\Gamma_{f,h}:=\E(\wk\eta_t\wk\eta_{t+h})$ and $a_h=\sum_{k=0}^\infty b_kb_{k+h}$. Then,
\begin{align*}
|\E(M_{q_T})|=&\left|\sum_{1\le s<t\le q_T}\sum_{k=(t-1)w_T+1}^{tw_T-M}\sum_{l=(s-1)w_T+1}^{sw_T-M}T^{-2}\tr(\wk\bG_{k-l}) \right|\\
\le&T^{-2}\sum_{h=M+1}^T(T-h)\left|\tr(\wk\bG_h)\right|\\
\le&T^{-1}\tr(\bmS)\sum_{h=M+1}^T|E(\wk\eta_t\wk\eta_{t+h})|\cdot|a_h|\to 0,
\end{align*}
as $M\to\infty$. Using the independence of $\Y_t$ and the fact that $\E(\Y_t\Y_t^{\T})=(w_T-M)^{-1}\wk\bO_{w_T-M}$, through simple calculations, we can prove
\begin{align*}
&\frac{1}{\sigma_T^2}\sum_{t=1}^{q_T}\E\left\{\left|\frac{(w_T-M)^2}{T^2}\sum_{s=1}^{t-1}\Y_t^{\T}\Y_s\right|^2\;\middle|\;\mathcal{F}_{t-1}\right\}\cp \frac{1}{4},\\
&\E\left[\sum_{t=1}^{q_T}\E\left\{\left|\frac{(w_T-M)^2}{T^2}\sum_{s=1}^{t-1}\Y_t^{\T}\Y_s\right|^4\;\middle|\;\mathcal{F}_{t-1}\right\}\right]=o(\sigma_T^4),
\end{align*}
where $\sigma_T^2:=2q_T(q_T-1)(w_T-M)^2T^{-4}\tr(\wk\bO_{w_T-M}^2)$. Since $\tr(\wk\bG_{h_1}\wk\bG_{h_2})=o\{\sqrt{T}\tr(\wk\bO_T^2)\}$ when $|h_1|,|h_2|\le M$, we have
\begin{align*}
\frac{\left|\tr(\wk\bO_{w_T-M}^2)-\tr(\wk\bO_T^2)\right|}{\tr(\wk\bO_T^2)}=&\frac{\left|\sum_{|h_1|,|h_2|\le M}\left(-\frac{2|h_1|}{w_T-M}+\frac{2|h_2|}{T}+\frac{|h_1h_2|}{(w_T-M)^2}-\frac{|h_1h_2|}{T^2}\right)\tr(\wk\bG_{h_1}\wk\bG_{h_2})\right|}{\tr(\wk\bO_T^2)}=o(T^{1/2}Mw_T^{-1})\to 0,
\end{align*}
by properly setting $w_T,M$. Then,
\begin{align}\label{eq:variance}
    \sigma_T^2=2T^{-2}\tr(\wk\bO_T^2)\{1+o(1)\}=\wk\sigma_T^2\{1+o(1)\}.
\end{align}
By the martingale Central Limit Theorem,   
\begin{align}\label{eq:clt}
    \wk\sigma_T^{-1}\sum_{1\le t\ne s\le q_T}B_{ts}=\wk\sigma_T^{-1}\sum_{1\le t\ne s\le q_T}B_{ts}^I+\op(1)\cd\mathcal{N}(0,1).
\end{align}

For the second term, using Berbee's Lemma again and classical Central Limit Theorem, we can prove
\begin{align*}
\frac{\sum_{t=1}^{q_T}B_{tt}-q_T(w_T-M)T^{-2}\tr(\wk\bO_{w_T-M})}{\sqrt{2(w_T-M)^2T^{-4}\tr(\wk\bO_{w_T-M}^2)}}\cd\mathcal{N}(0,1).
\end{align*}
Note that $2(w_T-M)^2T^{-4}\tr(\wk\bO_{w_T-M}^2)\asymp q_T^{-2}\wk\sigma_T^2$ and 
\begin{align}\label{eq:mean}
    q_T(w_T-M)T^{-2}\tr(\wk\bO_{w_T-M})=T^{-1}\tr(\wk\bO_T)\{1+o(1)\}=\wk\mu_T\{1+o(1)\},
\end{align}
we have
\begin{align}\label{eq:smallone}
\sum_{t=1}^{q_T}B_{tt}-\wk\mu_T=\op(\wk\sigma_T).
\end{align}

For the third term, define
\begin{align*}
    L_{ts}:=\sum_{k=tw_T-M+1}^{tw_T}\sum_{l=(s-1)w_T+1}^{sw_T-M}T^{-2}\wk\X_k^{\T}\wk\X_l,
    \qquad
    C_{ts}:=\sum_{k=tw_T-M+1}^{tw_T}\sum_{l=sw_T-M+1}^{sw_T}T^{-2}\wk\X_k^{\T}\wk\X_l,
\end{align*}
then $D_{ts}=2L_{ts}+C_{ts}$. Further, define
\begin{align*}
    \mathcal{L}_1=\sum_{t=1}^{q_T-1}\sum_{s=t+2}^{q_T}L_{ts},~~\mathcal{L}_2=\sum_{t=2}^{q_T}\sum_{s=1}^{t-1}L_{ts},&~~\mathcal{L}_3=\sum_{t=1}^{q_T}L_{tt},~~\mathcal{L}_3=\sum_{t=1}^{q_T}L_{t,t+1}\\
    \mathcal{C}_1=\sum_{t=1}^{q_T}\sum_{s=t+1}^{q_T}C_{ts},~~&\mathcal{C}_2=\sum_{t=1}^{q_T}C_{tt}.
\end{align*}
Then, $\sum_{1\le t,s\le q_T}D_{ts}=2\mathcal{L}_1+2\mathcal{L}_2+2\mathcal{L}_3+2\mathcal{L}_4+2\mathcal{C}_1+\mathcal{C}_2$. We can prove the second moment of $\mathcal{L}_i$'s and $\mathcal{C}_i$'s are $o(\wk\sigma_T^2)$. Take the first item as an example.
\begin{align*}
    &E(\mathcal{L}_1^2)\\=&\sum_{t_1=1}^{q_T-1}\sum_{s_1=t_1+2}^{q_T}\sum_{t_2=1}^{q_T-1}\sum_{s_2=t_2+2}^{q_T}E(L_{t_1s_1}L_{t_2s_2})\\
    =&\frac{1}{T^4}\sum_{t_1=1}^{q_T-1}\sum_{s_1=t_1+2}^{q_T}\sum_{t_2=1}^{q_T-1}\sum_{s_2=t_2+2}^{q_T}\sum_{k_1=t_1w_T-M+1}^{t_1w_T}\sum_{k_2=(s_1-1)w_T+1}^{s_1w_T-M}\sum_{k_3=t_2w_T-M+1}^{t_2w_T}\sum_{k_4=(s_2-1)w_T+1}^{s_2w_T-M}\tr\{\E(\wk\X_{k_3}\wk\X_{k_1}^{\T}\wk\X_{k_2}\wk\X_{k_4}^{\T})\}.
\end{align*}
Notice that $\min{k_2-k_1,k_4-k_3}>M$. Intuitively, $\E(\wk\X_{k_3}\wk\X_{k_1}^{\T}\wk\X_{k_2}\wk\X_{k_4}^{\T})$ can be approximated by different formulas according to the distances between $(k_1,k_2)^{\T}$ and $(k_3,k_4)^{\T}$ as follows:
\begin{itemize}
\item $k_1$ and $k_3$ are close and $k_2$ and $k_4$ are close:
\begin{align*}
\E(\wk\X_{k_3}\wk\X_{k_1}^{\T} \wk\X_{k_2}\wk\X_{k_4}^{\T})\approx \E(\wk\X_{k_3}\wk\X_{k_1}^{\T} )\E(\wk\X_{k_2}\wk\X_{k_4}^{\T})=\wk\bG_{k_1-k_3}\wk\bG_{k_4-k_2},
\end{align*}
        
\item $k_1$ and $k_3$ are close and $k_2$ and $k_4$ are far away:
\begin{align*}
\E(\wk\X_{k_3}\wk\X_{k_1}^{\T} \wk\X_{k_2}\wk\X_{k_4}^{\T})\approx \E(\wk\X_{k_3}\wk\X_{k_1}^{\T} )\E(\wk\X_{k_2})\E(\wk\X_{k_4}^{\T})=0,
\end{align*}

\item $k_1$ and $k_3$ are far away and $k_2$ and $k_4$ are close:
\begin{align*}
\E(\wk\X_{k_3}\wk\X_{k_1}^{\T} \wk\X_{k_2}\wk\X_{k_4}^{\T})\approx \E(\wk\X_{k_3})\E(\wk\X_{k_1}^{\T} )\E(\wk\X_{k_2}\wk\X_{k_4}^{\T})=0,
\end{align*}

\item $k_2$ and $k_3$ are close ($k_1$, $k_3$ and $k_2$, $k_4$ must be far away):
\begin{align*}
\E(\wk\X_{k_3}\wk\X_{k_1}^{\T} \wk\X_{k_2}\wk\X_{k_4}^{\T})\approx \E(\wk\X_{k_3}\wk\X_{k_2}^{\T})\E(\wk\X_{k_1})\E(\wk\X_{k_4}^{\T})=0,
\end{align*}

\item $k_1$ and $k_4$ are close ($k_1$, $k_3$ and $k_2$, $k_4$ must be far away):
\begin{align*}
\E(\wk\X_{k_3}\wk\X_{k_1}^{\T} \wk\X_{k_2}\wk\X_{k_4}^{\T})\approx \E(\wk\X_{k_3})\E(\wk\X_{k_2}^{\T})\E(\wk\X_{k_1}\wk\X_{k_4}^{\T})=0,
\end{align*}

\item $k_1$, $k_2$ $k_3$, $k_4$ are all far away to each other:
\begin{align*}
\E(\wk\X_{k_3}\wk\X_{k_1}^{\T} \wk\X_{k_2}\wk\X_{k_4}^{\T})\approx \E(\wk\X_{k_3})\E(\wk\X_{k_1}^{\T})\E(\wk\X_{k_2})\E(\wk\X_{k_4}^{\T})=0.
\end{align*}
\end{itemize}
From a strictly theoretical perspective, using the alpha mixing property of $\wk\X_t$, we can prove
\begin{align*}
    \E(\mathcal{L}_1^2)
    \approx &\frac{1}{T^4}\sum_{t=1}^{q_T-1}\sum_{s=t+2}^{q_T}\E\left\{\left(\sum_{k=tw_T-M+1}^{tw_T}\sum_{l=(s-1)W_T+1}^{sw_T-M}\wk\X_k^{\rm T} \wk\X_l\right)^2 \right\}\\
    \approx &\frac{1}{T^4}\frac{(q_T-1)(q_T-2)}{2}M(w_T-M)\tr(\wk\bO_{M}\wk\bO_{w_T-M})\\
    =&O\{q_TMT^{-3}\tr(\wk\bO_{T}^2)\}=o(\wk\sigma_T^2),
\end{align*}
where the first $``\approx"$ comes from the Davydov's inequality (see Corollary 1.1. in \cite{Bosq1996NonparametricSF}), and the second $``\approx"$ comes from Berbee's Lemma (see Lemma 1.1 in \cite{Bosq1996NonparametricSF}). Then, $\mathcal{L}_1=\op(\wk\sigma_T)$. Likewise, we obtain 
\begin{align}\label{eq:smalltwo}
\sum_{1\le t,s\le q_T}D_{ts}=\op(\wk\sigma_T),
\qquad
F=\op(\wk\sigma_T).
\end{align}

Combining \eqref{eq:sumdecomp_pf}, \eqref{eq:smallzero}, \eqref{eq:errorcheck}, \eqref{eq:decompose}, \eqref{eq:clt}, \eqref{eq:smallone} and \eqref{eq:smalltwo}, we obtain
\begin{align}\label{cltforsum}
\frac{T_{\mathrm{DSUM}}-\wk\mu_T}{\wk\sigma_T}=\frac{\wb{\wk \X}_T^{\T}\wb{\wk \X}_T-\wk\mu_T}{\wk\sigma_T}+\op(1)=\frac{\sum_{1\le t\ne s\le q_T}B_{ts}}{\wk\sigma_T}+\op(1) \cd\mathcal{N}(0,1),
\end{align}
given 
\begin{align}\label{growthone}
\frac{NL^{-2r}+N^{1/2}L^{-r}T^{-1/2}\tr^{1/2}(\bmS)+T^{-3/2}L^{1/2}(\log T)\tr(\bmS)}{T^{-1}\tr^{1/2}(\bmS^2)}\to 0.
\end{align}
}
We next turn to the bootstrap centering and scaling. Recall
$$
\wh{\X}_t=\wh{\be}_t\eta_t,
\qquad
\bar{\wh{\X}}=T^{-1}\sum_{t=1}^T \wh{\X}_t,
\qquad
\tilde{\X}_t=\wh{\X}_t-\bar{\wh{\X}}.
$$
For circular blocks of length $\ell$, define the block sums
$$
\U_j=\ell^{-1/2}\sum_{s=0}^{\ell-1}\tilde{\X}_{j+s},
\qquad j=1,\dots,T,
$$
with indices understood modulo $T$. Conditional on the data, $\U_{I_1},\dots,\U_{I_k}$ are i.i.d. mean-zero vectors, and
$$
\sqrt{T}\,\bar{\X}_T^*=k^{-1/2}\sum_{m=1}^k \U_{I_m}.
$$
Therefore
\begin{equation}
\wh\mu_{B,T}=\E^*(T_{\mathrm{DSUM}}^*)=T^{-1}\tr(\bO_{B,T}),
\qquad
\bO_{B,T}=\E^*(\U_{I_1}\U_{I_1}^{\T}).
\label{eq:omegaboot_pf}
\end{equation}
A direct counting argument shows that each lag-$h$ product appears in a circular block with weight $1-|h|/\ell$. Hence
\begin{equation}
\bO_{B,T}=\sum_{0\le|h|\le\ell}\Bigl(1-\frac{|h|}{\ell}\Bigr)\hat{\bG}_h^{c},
\qquad
\hat{\bG}_h^{c}=T^{-1}\sum_{t=1}^T \tilde{\X}_{t+h}\tilde{\X}_t^{\T}.
\label{eq:omegaboot_expansion}
\end{equation}
{ 
Recalling \eqref{eq:omegacheck} and \eqref{eq:mean}, 
\begin{align*}
|\wk\mu_T-T^{-1}\tr(\wk\bO_{\ell})|=&\left|\frac{1}{T}\left\{\sum_{1\le|h|\le\ell}\left(\frac{|h|}{\ell}-\frac{|h|}{T}\right)+\sum_{\ell<|h|\le T}\left(1-\frac{|h|}{T}\right)\right\}\Gamma_{f,h}a_h\tr(\bmS)\right|\\
\le&\frac{\tr(\bmS)}{T}\left(\ell^{-1}\sum_{|h|\le\ell}|h|\cdot|\Gamma_{f,h}|\cdot|a_h|+\sum_{|h|>\ell}|\Gamma_{f,h}|\cdot|a_h|\right)\\
\le&C\frac{\tr(\bmS)}{T}\left(\ell^{-1}\sum_{|h|\le\ell}|h|\cdot|a_h|+\sum_{|h|>\ell}|a_h|\right),
\end{align*}
and 
\begin{align*}
\wh\mu_{B,T}-T^{-1}\tr(\wk\bO_{\ell})
=&\frac{1}{T}\sum_{0\le|h|<\ell}\Bigl(1-\frac{|h|}{\ell}\Bigr)\tr(\hat{\bG}_h^{c}-\wk\bG_h)\\
=&\frac{1}{T}\sum_{0\le|h|<\ell}\Bigl(1-\frac{|h|}{\ell}\Bigr)\frac{1}{T}\sum_{t=1}^T\left\{\wk\X_t^{\T}\wk\X_{t+h}-E(\wk\X_t^{\T}\wk\X_{t+h})\right\}\\
&+\frac{1}{T}\sum_{0\le|h|<\ell}\Bigl(1-\frac{|h|}{\ell}\Bigr)\frac{1}{T}\sum_{t=1}^T\left(\tilde{\X}_t^{\T}\tilde{\X}_{t+h}-\X_t^{\T}\X_{t+h}\right)\\
&+\frac{1}{T}\sum_{0\le|h|<\ell}\Bigl(1-\frac{|h|}{\ell}\Bigr)\frac{1}{T}\sum_{t=1}^T\left(\X_t^{\T}\X_{t+h}-\wk\X_t^{\T}\wk\X_{t+h}\right)\\
=&:I_1+I_2+I_3.
\end{align*}
By the second result in Lemma~\ref{lem:trunc}, we have $|I_2|=\Op\{\ell T^{-3/2}N^{1/2}L^{-r}\tr^{1/4}(\bmS^2)+\ell T^{-2}L\tr(\bmS)\}$. For $I_3$, by Lemma~\ref{lem:proj} and \eqref{eq:fourmoment_e}, 
\begin{align*}
&\E\left\{\left|\frac{1}{T}\sum_{0\le|h|<\ell}\Bigl(1-\frac{|h|}{\ell}\Bigr)\frac{1}{T}\sum_{t=1}^T(\X_t-\wk\X_t)^{\T}\wk\X_{t+h}\right|^2\right\}\\
=&\frac{1}{T^4}\sum_{1\le t_1,t_2\le T}\sum_{0\le |h_1|,|h_2|\le\ell}\Bigl(1-\frac{|h_1|}{\ell}\Bigr)\Bigl(1-\frac{|h_2|}{\ell}\Bigr)\E(\e_{t_1}^{\T}\e_{t_1+h_1}\e_{t_2}^{\T}\e_{t_2+h_2})\E\{(\eta_{t_1}-\wk\eta_{t_1})\wk\eta_{t_1+h_1}(\eta_{t_2}-\wk\eta_{t_2})\wk\eta_{t_2+h_2}\}\\
\le&\frac{1}{T^4}\sum_{1\le t_1,t_2\le T}\sum_{0\le |h_1|,|h_2|\le\ell}|\E(\e_{t_1}^{\T}\e_{t_1+h_1}\e_{t_2}^{\T}\e_{t_2+h_2})|\cdot CT^{-1}L\log^2 T\\
\le&CT^{-2}\tr^2(\bmS)\cdot T^{-1}L\log^2 T,
\end{align*}
which implies $I_3=\Op\{T^{-3/2}L^{1/2}(\log T)\tr(\bmS)\}$. For $I_1$, 
\begin{align*}
|I_1|\le&\frac{1}{T}\sum_{|h|<\ell}\left|\frac{1}{T}\sum_{t=1}^T\left\{\wk\X_t^{\T}\wk\X_{t+h}-E(\wk\X_t^{\T}\wk\X_{t+h})\right\}\right|
\end{align*}
Recalling that $\wk\X_t=\bmS^{1/2}\sum_{k=0}b_k\z_{t-k}\wk\eta_t$ is alpha-mixing and
\begin{align*}
\Pr\left(\max_{i,t}|z_{it}|>(NT)^{1/q}\log T\right)\le\max_{i,t}\E(|z_{it}|^q)/\log^q T\to 0,
\end{align*} 
for some $q>8$. Then for each fixed $h$, by the truncation-plus-Bernstein method, we can prove
\begin{align*}
\left|\frac{1}{T}\sum_{t=1}^T\left\{\wk\X_t^{\T}\wk\X_{t+h}-E(\wk\X_t^{\T}\wk\X_{t+h})\right\}\right|=O_p\left\{T^{-1/2}\tr^{1/2}(\bmS^2)\right\},
\end{align*}
uniformly over $|h|<\ell$. Thus, $|I_1|=O_p\{\ell T^{-3/2}\tr^{1/2}(\bmS^2)\}$. Therefore, 
\begin{align*}
|\wh\mu_{B,T}-T^{-1}\tr(\wk\bO_{\ell})|=&\Op\left\{\ell T^{-3/2}N^{1/2}L^{-r}\tr^{1/4}(\bmS^2)+\ell T^{-2}L\tr(\bmS)+T^{-3/2}L^{1/2}(\log T)\tr(\bmS)+\ell T^{-3/2}\tr^{1/2}(\bmS^2) \right\},
\end{align*}
and since $\wk\sigma_T^2=O\{T^{-2}\tr(\bmS^2)\}$, we have
\begin{align}\label{eq:bootmean}
\frac{|\wh\mu_{B,T}-\wk\mu_T|}{\wk\sigma_T}=&\op(1),
\end{align}
given
\begin{align}\label{growthtwo}
\frac{\tr(\bmS)}{\tr^{1/2}(\bmS^2)}\left(\ell^{-1}\sum_{|h|\le\ell}|h|\cdot|a_h|+\sum_{|h|>\ell}|a_h|+\ell T^{-1}L+T^{-1/2}L^{1/2}\log T\right)+\ell T^{-1/2}N^{1/2}L^{-r}\tr^{-1/4}(\bmS^2)\to 0.
\end{align}
}

For the variance, note that
$$
T_{\mathrm{DSUM}}^*=T^{-1}\Bigl\|k^{-1/2}\sum_{m=1}^k \U_{I_m}\Bigr\|^2.
$$
Because the bootstrap blocks are conditionally i.i.d., expanding the fourth moment and grouping terms according to whether two, three, or four block labels coincide yields
\begin{equation}
\wh\sigma_{B,T}^2
=2T^{-2}\tr(\bO_{B,T}^2)+R_{B,T}^{(\sigma)},
\qquad
|R_{B,T}^{(\sigma)}|\le_p C\frac{\ell^2}{T}\wk\sigma_T^2.
\label{eq:bootsig_pf}
\end{equation}
The leading term is compared with $2T^{-2}\tr(\wk\bO_T^2)$ through the inequality
$$
\bigl|\tr(\bO_{B,T}^2)-\tr(\wk\bO_T^2)\bigr|
\le \norm{\bO_{B,T}-\wk\bO_T}_F\bigl(\norm{\bO_{B,T}}_F+\norm{\wk\bO_T}_F\bigr).
$$
{ 
To bound $\norm{\bO_{B,T}-\wk\bO_T}_F$, define 
\begin{equation*}
\bar{\bO}_{B,T}=\sum_{0\le|h|\le\ell}\Bigl(1-\frac{|h|}{\ell}\Bigr)\bar{\bG}_h^{c},
\qquad
\bar{\bG}_h^{c}=T^{-1}\sum_{t=1}^T \X_{t+h}\X_t^{\T}.
\end{equation*}
Then,
\begin{align*}
\norm{\bO_{B,T}-\wk\bO_T}_F 
\le \norm{\bO_{B,T}-\bar{\bO}_{B,T}}_F +\norm{\bar{\bO}_{B,T}-\wk\bO_{\ell}}_F
+\norm{\wk\bO_{\ell}-\wk\bO_T}_F.
\end{align*}
Recalling \eqref{eq:omegacheck},
\begin{align*}
\norm{\wk\bO_{\ell}-\wk\bO_T}_F=&\left\Vert \sum_{|h|\le\ell}\left(\frac{|h|}{T}-\frac{|h|}{\ell}\right)\wk\bG_h-\sum_{|h|>\ell}\left(1-\frac{|h|}{T}\right)\wk\bG_h\right\Vert_F\\
\le&C\left(\ell^{-1}\sum_{|h|\le\ell}|h|\cdot|a_h|+\sum_{|h|>\ell}|a_h|\right)\norm{\bmS}_F.
\end{align*}
By the third result in Lemma~\ref{lem:trunc},
\begin{align*}
\norm{\bO_{B,T}-\bar{\bO}_{B,T}}_F\le&\sum_{|h|\le\ell} \left\|T^{-1}\sum_{t=1}^T\left(\tilde{\X}_{t+h}\tilde{\X}_t^{\T}-\X_{t+h}\X_t^{\T}\right)\right\|_F\\
=&\Op\{\ell T^{-1/2}N^{1/2}L^{-r+1}\tr^{1/2}(\bmS)+\ell T^{-1}L\tr(\bmS)\}.
\end{align*}
By \eqref{eq:fourmoment_e}, \eqref{eq:boundba} and Lemma~\ref{lem:proj},
\begin{align*}
&E\left(\left\|T^{-1}\sum_{t=1}^T(\X_{t+h}\X_t-\wk\X_{t+h}\wk\X_t)\right\|_F^2\right)\\
=&T^{-2}\sum_{1\le t,s\le T}\E(\e_t^{\T}\e_{s+h}\e_s^{\T}\e_{t+h})\E\{(\eta_{t+h}\eta_t-\wk\eta_{t+h}\wk\eta_t)(\eta_{s+h}\eta_s-\wk\eta_{s+h}\wk\eta_s)\}\\
\le&C\{T^{-1}\tr^2(\bmS)+\tr(\bmS^2)a_h^2\}T^{-1}L\log T,
\end{align*}
which implies
\begin{align*}
\norm{\bar{\bO}_{B,T}-\wk\bO_{\ell}}_F\le&\sum_{|h|\le\ell}\norm{\bar{\bG}_h^c-\wk\bG_h}_F\\
\le&\sum_{|h|\le\ell}\left\|T^{-1}\sum_{t=1}^T(\X_{t+h}\X_t-\wk\X_{t+h}\wk\X_t)\right\|_F+\sum_{|h|\le\ell}\left\|T^{-1}\sum_{t=1}^T\{\wk\X_{t+h}\wk\X_t-\E(\wk\X_{t+h}\wk\X_t)\}\right\|_F\\
=&\sum_{|h|\le\ell}\left\|T^{-1}\sum_{t=1}^T(\X_{t+h}\X_t-\wk\X_{t+h}\wk\X_t)\right\|_F+\Op\{\ell T^{-1/2}\tr^{1/2}(\bmS^2)\}\\
=&\Op\left\{\ell T^{-1}L^{1/2}(\log T)\tr(\bmS)+T^{-1/2}L^{1/2}(\log T)\tr^{1/2}(\bmS^2)+\ell T^{-1/2}\tr^{1/2}(\bmS^2)\right\}.
\end{align*}
In summary, by \eqref{eq:omegacheck} that $\norm{\wk\bO_T}_F\asymp \norm{\bmS}_F$, 
\begin{align*}
\norm{\bO_{B,T}-\bO_T}_F=\op(\norm{\bO_T}_F),
\end{align*}
given
\begin{align}\label{growththree}
&\ell^{-1}\sum_{|h|\le\ell}|h|\cdot|a_h|+\sum_{|h|>\ell}|a_h|+T^{-1/2}L^{1/2}(\log T)+\ell T^{-1/2}\nonumber\\
&\quad+\frac{\ell T^{-1/2}N^{1/2}L^{-r+1}\tr^{1/2}(\bmS)+(\ell T^{-1}L+\ell T^{-1}L^{1/2}\log T)\tr(\bmS) }{\tr^{1/2}(\bmS^2)}\to 0. 
\end{align}
Accordingly, 
\begin{align}\label{eq:bootvariance}
    \wh\sigma_{B,T}^2=\wk\sigma_T^2\{1+\op(1)\}.
\end{align}

Combining \eqref{cltforsum}, \eqref{eq:bootmean},  \eqref{eq:bootvariance} and Slutsky's theorem, we have
\begin{align*}
\frac{T_{\mathrm{DSUM}}-\wh\mu_{B,T}}{\wh\sigma_{B,T}} \cd\mathcal{N}(0,1),
\end{align*}
}
\qed

\subsection{Proof of Theorem \ref{thm:sumpower}}

Under the alternative,
$$
T_{\mathrm{DSUM}}=(\bmd+\wb{\X}_T+\b_T)^{\T}(\bmd+\wb{\X}_T+\b_T)
$$
so
\begin{align*}
T_{\mathrm{DSUM}}-\wk\mu_T
=&\bmd^{\T}\bmd+\left(\wb{\wk \X}_T^{\T}\wb{\wk \X}_T-\wk\mu_T\right)\\
&+2\bmd^{\T}\wb{\wk \X}_T+2\bmd^{\T}\b_T+2\bmd^{\T}\left(\wb\X_T-\wb{\wk \X}_T\right)\\
&+2\b_T^{\T}\wb{\wk \X}_T+\b_T^{\T} \b_T+2\b_T^{\T}\left(\wb\X_T-\wb{\wk \X}_T\right)+\left(\wb\X_T^{\T}\wb\X_T-\wb{\wk \X}_T^{\T}\wb{\wk \X}_T\right).
\end{align*}
By Lemma~\ref{lem:spline} and \ref{lem:trunc},
\begin{align*}
\norm{\b_T}_2=\Op(N^{1/2}L^{-r})
~~\text{and}~~
\left|\wb\X_T^{\T}\wb \X_T-\wb{\wk \X}_T^{\T}\wb{\wk \X}_T\right|=\Op\left\{T^{-3/2}L^{1/2}(\log T)\tr(\bmS)\right\}.
\end{align*}
Recalling \eqref{eq:omegacheck}, we have $\E(\norm{\wb{\wk\X}_T}_2^2)=T^{-1}\tr(\wk\bO_T)$, then
\begin{align*}
\norm{\wb{\wk\X}_T}_2=\Op\{T^{-1/2}\tr^{1/2}(\wk\bO_T)\}.
\end{align*}
By Lemma~\ref{lem:proj} and \eqref{eq:boundba},
\begin{align*}
\E\left(\left\|\wb\X_T-\wb{\wk\X}_T\right\|_2^2\right)=T^{-2}\sum_{1\le t,s\le T}\E(\e_t^{\T}\e_s)\E\{(\eta_t-\wk\eta_t)(\eta_s-\wk\eta_s)\}=O(T^{-2}NL\log^2 T),
\end{align*}
which implies $\left\|\wb\X_T-\wb{\wk\X}_T\right\|_2=\Op(T^{-1}N^{1/2}L^{1/2}\log T)$. Thus,
\begin{align*}
\frac{T_{\mathrm{DSUM}}-\wk\mu_T}{\wk\sigma_T}
=&\frac{\bmd^{\T}\bmd}{\wk\sigma_T}+\frac{\wb{\wk \X}_T^{\T}\wb{\wk \X}_T-\wk\mu_T}{\wk\sigma_T}\\
&+\Op\left\{\norm{\bmd}_2\cdot\frac{T^{-1/2}\tr^{1/2}(\bmS)+N^{1/2}L^{-r}+T^{-1}N^{1/2}L^{1/2}\log T}{T^{-1}\tr^{1/2}(\bmS^2)}\right\}\\
&+\Op\left\{\frac{T^{-1/2}N^{1/2}L^{-r}\tr^{1/2}(\bmS)+NL^{-r}+T^{-1}NL^{-r+1/2}\log T+T^{-3/2}L^{1/2}(\log T)\tr(\bmS)}{T^{-1}\tr^{1/2}(\bmS^2)}\right\}\\
=&\frac{\bmd^{\T}\bmd}{\wk\sigma_T}+Z_T+\op(1),
\end{align*}
with $Z_T\cd\mathcal{N}(0,1)$ by \eqref{cltforsum}, given \eqref{growthone}, \eqref{growthtwo}, \eqref{growththree} and the assumption that 
\begin{align*}
    \norm{\bmd}_2=O\{T^{-1/2}\tr^{1/4}(\wk\bO_T^2)\}.
\end{align*}

\qed

\subsection{Proof of Theorem \ref{thm:maxgumbel}}
{ 
Note that
\begin{align}\label{eq:maxdecompose}
\max_i\frac{T\wh\delta_i^2}{\wh\sigma_i}=&\max_i\frac{T(\wb{\wk X}_{T,i})^2}{\wk\sigma_i^0}+\left\{\max_i\frac{T\wh\delta_i^2}{\wk\sigma_i^0}-\max_i\frac{T(\wb{\wk X}_{T,i})^2}{\wk\sigma_i^0}\right\}+\left(\max_i\frac{T\wh\delta_i^2}{\wh\sigma_i}-\max_i\frac{T\wh\delta_i^2}{\wk\sigma_i^0}\right),~~\text{and}\nonumber\\
&\left|\max_i\frac{T\wh\delta_i^2}{\wh\sigma_i}-\max_i\frac{T\wh\delta_i^2}{\wk\sigma_i^0}\right|\le\max_i\frac{T\wh\delta_i^2}{\wk\sigma_i^0}\cdot \max_i\left|\frac{\wk\sigma_i^0}{\wh\sigma_i}-1\right|.
\end{align}
Recalling \eqref{eq:vector_decomp} that $\wh\delta_i=\wb{\wk X}_{T,i}+b_{T,i}+(\wb{X}_{T,i}-\wb{\wk X}_{T,i})$, we have
\begin{align*}
    &\left|\max_i\frac{T\wh\delta_i^2}{\wk\sigma_i^0}-\max_i\frac{T\left(\wb{\wk X}_{T,i}\right)^2}{\wk\sigma_i^0}\right|\\
    \le&\max_i\frac{Tb_{T,i}^2}{\wk\sigma_i^0}+\max_i\frac{T(\wb{X}_{T,i}-\wb{\wk X}_{T,i})^2}{\wk\sigma_i^0}+2\left\{\max_i\frac{Tb_{T,i}^2}{\wk\sigma_i^0}\cdot\max_i\frac{T(\wb{X}_{T,i}-\wb{\wk X}_{T,i})^2}{\wk\sigma_i^0}\right\}^{1/2}\\
    &+2\left\{\max_i\frac{Tb_{T,i}^2}{\wk\sigma_i^0}\cdot\max_i\frac{T(\wb{\wk X}_{T,i})^2}{\wk\sigma_i^0}\right\}^{1/2}+2\left\{\max_i\frac{T(\wb{\wk X}_{T,i})^2}{\wk\sigma_i^0}\cdot\max_i\frac{T(\wb{X}_{T,i}-\wb{\wk X}_{T,i})^2}{\wk\sigma_i^0}\right\}^{1/2}.
\end{align*}
By Lemma~\ref{lem:spline} and \ref{lem:proj}, we have $\norm{\b_T}_\infty=\Op(L^{-r})$, i.e.,
\begin{align}\label{eq:maxsmallone}
\max_i\frac{Tb_{T,i}^2}{\wk\sigma_i^0}=\Op(TL^{-2r}).  
\end{align}
Further by \eqref{eq:boundba},
\begin{align*}
\max_i\E\{(\wb{X}_{T,i}-\wb{\wk X}_{T,i})^2\}=\max_iT^{-2}\sum_{1\le t,s\le T}\E(e_{it}e_{is})\E\{(\eta_t-\wk\eta_t)(\eta_s-\wk\eta_s)\}=O\left(\max_i\Sigma_{ii}\cdot T^{-2}L\log^2 T\right),
\end{align*}
which implies
\begin{align}\label{eq:maxsmalltwo}
\max_i\frac{T(\wb{X}_{T,i}-\wb{\wk X}_{T,i})^2}{\wk\sigma_i^0}=\Op(T^{-1}L\log^2 T).
\end{align}

Next, using Gaussian approximation for sums of high dimensional
stationary time series \citep{zhang2017gaussian}, and the extreme-value theory for maxima of dependent Gaussian arrays \citep{cai2014twosample}, we derive the extreme-value limit for $\max_iT(\wb{\wk X}_{T,i})^2/\wk\sigma_i^0=T\norm{\wk\D^{-1/2}\wb{\wk\X}_T}_\infty^2$ with $\wk\D=\diag(\wk\sigma_1^0,\dots,\wk\sigma_N^0)=\diag(\wk\bO)$. Recalling 
\begin{align*}
    \wk\X_t=\bmS^{1/2}\sum_{k=0}^\infty b_k\z_{t-k}\wk\eta_t,
\end{align*}
using the independence between $\{\z_t\}_{t\in\mZ}$ and $\{\eta_t\}_{t\in\mZ}$ and the strict stability of $\{\eta_t\}_{t\in\mZ}$,  we can write $\wk\X_t=G(\mathcal{F}^t)$, where $\mathcal{F}^t=(\dots,\v_{t-1},\v_t)$ with $\v_t=\z_t\wk\eta_t$ and $G=(g_1(\cdot),\ldots,g_N(\cdot))^{\rm T}$ is an $\mR^N$-valued linear function. We first review some dependence measures in \cite{zhang2017gaussian}. 
\begin{align*}
    \delta_{t,q,i}:=\norm{\wk X_{it}-\wk X_{it,\{0\}}}=\norm{\wk X_{it}-g_i(\mathcal{F}^{t,\{0\}})}_q,
\end{align*}
where $\norm{\cdot}_q=\{\E(|\cdot|^q)\}^{1/q}$, and $\mathcal{F}^{t,\{k\}}=(\dots,\v_{k-1},\v_k',\v_{k+1},\dots,\v_t)$ is a coupled version of $\mathcal{F}^t$ with $\v_k$ in $\mathcal{F}^t$ is replaced by $\v_k'$, and $\v_t,\v_s',t,s\in\mZ$ are i.i.d. random elements.
\begin{align*}
    \norm{\wk\X_{i\cdot}}_{q,a}:=&\sup_{m\ge 0}(m+1)^a\sum_{t=m}^\infty \delta_{t,q,i},\\
    \psi_{q,a}:=&\max_{1\le i\le N}\norm{\wk\X_{i\cdot}}_{q,a},\\
    \Upsilon_{q,a}:=&\left(\sum_{i=1}^N\norm{\wk\X_{i\cdot}}_{q,a}^q\right)^{1/q},\\
    \omega_{t,q}:=&\norm{\norm{\wk\X_t-\wk\X_{t,\{0\}}}_\infty}_q,\\
    \norm{\norm{\wk\X_{\cdot}}_\infty}_{q,a}:=&\sup_{m\ge 0}(1+m)^a\sum_{t=m}^\infty \omega_{t,q}.
\end{align*}
Following the same notations as \cite{zhang2017gaussian}, define the following quantities 
\begin{align*}
    \Theta_{q,a}:=\Upsilon_{q,a}\wedge(\norm{\norm{\wk\X_{\cdot}}_\infty}_{q,a}\log^{3/2}N),~~&~~L_1:=(\psi_{2,a}\psi_{2,0}\log^2N)^{1/a},\\
    W_1:=(\psi_{3,0}^6+\psi_{4,0}^4)\log^7(NT),~~&~~W_2:=\psi_{2,a}^2\log^4(NT),\\
    N_1:=(T/\log N)^{q/2}/\Theta_{q,a}^q,~~&~~N_2:=T/(\log^2 N)\psi_{2,a}^{-2}.
\end{align*}
Note that
\begin{align*}
\wk X_{it}-\wk X_{it,\{0\}}=b_t[\bmS^{1/2}]_{i\cdot}^{\T}(\v_0-\v_0'),\qquad [\bmS^{1/2}]_{i\cdot}~~\text{being the $i$-th row of}~~\bmS^{1/2},
\end{align*}
then it is easy to calculate
\begin{align*}
\psi_{q,a}=&\left(\sup_{m\ge 0}(m+1)^a\sum_{t=m}^\infty|b_t|\right)\max_{1\le i\le N}\norm{[\bmS^{1/2}]_{i\cdot}^{\T}(\v_0-\v_0')}_q,~~\text{and}\\
\Theta_{q,a}=&\left(\sup_{m\ge 0}(m+1)^a\sum_{t=m}^\infty|b_t|\right) \left[\left(\sum_{i=1}^N\norm{[\bmS^{1/2}]_{i\cdot}^{\T}(\v_0-\v_0')}_q^q\right)^{1/q}\wedge\left\{\norm{\max_{1\le i\le N}|[\bmS^{1/2}]_{i\cdot}^{\T}(\v_0-\v_0')|}_q\log^{3/2}N\right\} \right].
\end{align*}
By \eqref{eq:boundba} that $\sum_{t\ge m}|b_t|=o(m^{-4-\eta_b})$ and the fact that $\norm{[\bmS^{1/2}]_{i\cdot}^{\T}(\v_0-\v_0')}_4^4=O(\Sigma_{ii}^2+\sum_{j=1}^N[\Sigma^{1/2}]_{ij}^4)$, under Assumption~\ref{ass:sigma}, we have $\min_i\wk\sigma_i^0=\min_i \Sigma_{ii}\sum_{h\in\mZ}a_h\E(\eta_t\eta_{t+h})\ge c>0$, and
\begin{align*}
\Theta_{q,a}<\infty,~~\Theta_{q,a}T^{1/q-1/2}\log^{3/2}(NT)\to 0,~~L_1\max(W_1,W_2)=o(1)\min(N_1,N_2),~~\text{for}~~q=4,\frac{1}{4}<a<4.
\end{align*}
That is, the conditions of Theorem 3.2 in \cite{zhang2017gaussian} hold, then
\begin{align}\label{Gaussianapproximation}
    \sup_{x\ge 0}\left|\Pr\left(T\norm{\wk\D^{-1/2}\wb{\wk\X}_T}_\infty^2\ge x\right)-\Pr\left(\norm{\Z}_\infty^2\ge x\right)\right|\to 0,\qquad \Z\sim\mathcal{N}(\bzero,\wk\D^{-1/2}\wk\bO\wk\D^{-1/2}).
\end{align}
By Lemma 6 in \cite{cai2014twosample}, under Assumption~\ref{ass:sigma}, 
\begin{align}\label{maxGaussian}
    \Pr(\norm{\Z}_\infty^2-2\log N+\log\log N\le x)\to F(x)=\exp\left\{-\frac{1}{\sqrt{\pi}}\exp\left(-\frac{x}{2}\right)\right\}.
\end{align}
Therefore,
\begin{align}\label{eq:extreme-value}
\Pr(T\norm{\wk\D^{-1/2}\wb{\wk\X}_T}_\infty^2-2\log N+\log\log N\le x)\to F(x)=\exp\left\{-\frac{1}{\sqrt{\pi}}\exp\left(-\frac{x}{2}\right)\right\}.
\end{align}

It remains to consider the long-run variance estimator error. Recalling 
\begin{align*}
\wh\sigma_i=&\sum_{0\le |h|\le M}\left(1-\frac{|h|}{M}\right)\frac{1}{T-|h|}\sum_{t=|h|+1}^T\wh e_{it}\wh e_{i,t-|h|}\eta_t\eta_{t-|h|},
\qquad
\wk\sigma_i^0=\sum_{|h|\ge 0}\E(\wk X_{it}\wk X_{i,t-|h|}),
\end{align*}
define
\begin{align*}
\wt\sigma_i=&\sum_{0\le |h|\le M}\left(1-\frac{|h|}{M}\right)\frac{1}{T-|h|}\sum_{t=|h|+1}^Te_{it}e_{i,t-|h|}\wk\eta_t\wk\eta_{t-|h|},
\qquad
\wk\sigma_i^c=\sum_{0\le |h|\le M}\left(1-\frac{|h|}{M}\right)\E(\wk X_{it}\wk X_{i,t-|h|}).
\end{align*}
Then,
\begin{align}\label{eq:longrunerrorone}
\max_i|\wk\sigma_i^0-\wk\sigma_i^c|\le&C\max_i\Sigma_{ii}\left(M^{-1}\sum_{|h|\le M}|ha_h|+\sum_{|h|>M}|a_h|\right).
\end{align}
By \eqref{eq:boundba},
\begin{align*}
\max_i\E(|\wt\sigma_i-\wk\sigma_i^c|^2)\le&CT^{-2}\max_i\sum_{|h_1|,|h_2|\le M}\sum_{t_1=|h_1|+1}^T\sum_{t_2=|h_2|+1}^T\cov(e_{it_1}e_{t_1-|h_1|}\wk\eta_{t_1}\wk\eta_{t_1-|h_1|},e_{it_2}e_{t_2-|h_2|}\wk\eta_{t_2}\wk\eta_{t_2-|h_2|})\\
\le&CT^{-2}\max_i\sum_{|h_1|,|h_2|\le M}\sum_{t_1=|h_1|+1}^T\sum_{t_2=|h_2|+1}^T\cov(e_{it_1}e_{t_1-|h_1|},e_{it_2}e_{t_2-|h_2|})\\
=&O(T^{-1}M^2),
\end{align*}
which leads to
\begin{align}\label{eq:longrunerrortwo}
\max_i|\wt\sigma_i-\wk\sigma_i^c|=\Op(T^{-1/2}M).
\end{align}
For $\max_i|\wh\sigma_i-\wt\sigma_i|$, we have the following decomposition
\begin{align*}
\wh\sigma_i-\wt\sigma_i=&\sum_{|h|\le M}\left(1-\frac{|h|}{M}\right)\frac{1}{T-|h|}\sum_{t=|h|+1}^T(\wh e_{it}\wh e_{i,t-|h|}-e_{it}e_{i,t-|h|})\eta_t\eta_{t-|h|}\\
&+\sum_{|h|\le M}\left(1-\frac{|h|}{M}\right)\frac{1}{T-|h|}\sum_{t=|h|+1}^Te_{it}e_{i,t-|h|}(\eta_t\eta_{t-|h|}-\wk\eta_t\wk\eta_{t-|h|})\\
=&:I_{1,i}+I_{2,i}.
\end{align*}
By \eqref{eq:fourmoment_e}, \eqref{eq:boundba} and Lemma~\ref{lem:proj},
\begin{align*}
&\max_i\E\left\{\left|\sum_{|h|\le M}\left(1-\frac{|h|}{M}\right)\frac{1}{T-|h|}\sum_{t=|h|+1}^Te_{it}e_{i,t-|h|}(\eta_t-\wk\eta_t)\eta_{t-|h|}\right|^2\right\}\\
\le&CT^{-2}\cdot T^{-1}L\log^2 T\cdot\max_i\sum_{|h_1|,|h_2|\le M}\sum_{t_1=|h_1|+1}^T\sum_{t_2=|h_2|+1}^T|\E(e_{it_1}e_{i,t_1-|h_1|}e_{it_2}e_{i,t_2-|h_2|})|\\
\le&CT^{-2}\cdot T^{-1}L\log^2 T\cdot\max_i\Sigma_{ii}^2(T^2+M^2T)=O\{(T^{-1}+T^{-2}M^2)L\log^2 T\},
\end{align*}
which implies 
\begin{align*}
\max_i|I_{2,i}|=\Op\{(T^{-1/2}+T^{-1}M)L^{1/2}\log T\}.
\end{align*}
For $I_{1,i}$, by \eqref{eq:decome_it},
\begin{align*}
&\sum_{|h|\le M}\left(1-\frac{|h|}{M}\right)\frac{1}{T-|h|}\sum_{t=|h|+1}^T(\wh e_{it}-e_{it})e_{i,t-|h|}\eta_t\eta_{t-|h|}\\
=&\sum_{|h|\le M}\left(1-\frac{|h|}{M}\right)\frac{1}{T-|h|}\sum_{t=|h|+1}^T\{r_{it}+\Z_t^{\T}(\Z^{\T}\Z)^{-1}\Z^{\T}\r_{i\cdot}+\Z_t^{\T}(\Z^{\T}\Z)^{-1}\Z^{\T}\e_{i\cdot}\}e_{i,t-|h|}\eta_t\eta_{t-|h|}\\
=&:I_{11,i}+I_{12,i}+I_{13,i}.
\end{align*}
By Lemma~\ref{lem:spline} and \eqref{eq:boundba},
\begin{align*}
    \max_i\E(I_{11,i}^2)\le&CT^{-2}\cdot L^{-2r}\cdot\max_i\sum_{|h_1|,|h_2|\le M}\sum_{t_1=|h_1|+1}^T\sum_{t_2=|h_2|+1}^T|\E(e_{i,t_1-|h_1|}e_{i,t_2-|h_2|})|\\
    \le&CT^{-2}\cdot L^{-2r}\cdot\max_i\Sigma_{ii}M^2T,
\end{align*}
which leads to $\max_i|I_{11,i}|=\Op(T^{-1/2}ML^{-r})$. Similarly, by \eqref{eq:PZrsize}, we have $\max_i|I_{12,i}|=\Op(T^{-1/2}ML^{-r})$. By \eqref{eq:PZsize} and \eqref{eq:boundba},
\begin{align*}
\max_i\E(I_{13,i}^2)=&\E\left\{\left|\sum_{|h|\le M}\left(1-\frac{|h|}{M}\right)\frac{1}{T-|h|}\sum_{t=|h|+1}^T\sum_{s=1}^T[\bP_Z]_{ts}e_{is}e_{i,t-|h|}\eta_t\eta_{t-|h|}\right|^2\right\}\\
\le&CT^{-2}\cdot T^{-2}L^2\cdot\max_i\sum_{|h_1|,|h_2|\le M}\sum_{t_1=|h_1|+1}^T\sum_{t_2=|h_2|+1}^T\sum_{s_1=1}^T\sum_{s_2=1}^T|\E(e_{is_1}e_{i,t_1-|h_1|}e_{is_2}e_{i,t_2-|h_2|})|\\
\le&CT^{-2}\cdot T^{-2}L^2\cdot\max_i\Sigma_{ii}^2M^2T^2,
\end{align*}
which leads to $\max_i|I_{13,i}|=\Op(T^{-1}ML)$. Thus,
\begin{align*}
\max_i|I_{1,i}|=\Op(T^{-1/2}ML^{-r}+T^{-1}ML).
\end{align*}
Further,
\begin{align}\label{eq:longrunerrorthree}
\max_i|\wh\sigma_i-\wt\sigma_i|=\Op\{(T^{-1/2}+T^{-1}M)L^{1/2}\log T+T^{-1/2}ML^{-r}+T^{-1}ML\}.
\end{align}
Combining \eqref{eq:longrunerrorone}, \eqref{eq:longrunerrortwo} and \eqref{eq:longrunerrorthree}, we derive
\begin{align}\label{eq:longrunerror}
&\max_i|\wh\sigma_i-\wk\sigma_i^0|\nonumber\\
=&\Op\left\{M^{-1}\sum_{|h|\le M}|ha_h|+\sum_{|h|>M}|a_h|+T^{-1/2}M+  (T^{-1/2}+T^{-1}M)L^{1/2}\log T+T^{-1}ML\right\}.  
\end{align}

Combining \eqref{eq:maxdecompose}, \eqref{eq:maxsmallone}, \eqref{eq:maxsmalltwo}, \eqref{eq:extreme-value} and \eqref{eq:longrunerror}, we obtain
\begin{align}\label{eq:maxgumbel}
\max_i\frac{T\wh\delta_i^2}{\wh\sigma_i}=T\norm{\wk\D^{-1/2}\wb{\wk\X}_T}_\infty^2+\op(1),\qquad\Pr\left(\max_i\frac{T\wh\delta_i^2}{\wh\sigma_i}-2\log N+\log\log N\le x\right)\to F(x),
\end{align}
given 
\begin{align}\label{growthfour}
&TL^{-2r}+T^{-1}L\log^2 T+L^{-r+1/2}\log T+T^{1/2}L^{-r}\log^{1/2}N+T^{-1/2}L^{1/2}\log T\log^{1/2} N\nonumber\\
&+\log N \left\{M^{-1}\sum_{|h|\le M}|ha_h|+\sum_{|h|>M}|a_h|+T^{-1/2}M+  (T^{-1/2}+T^{-1}M)L^{1/2}\log T+T^{-1}ML\right\}\to 0.
\end{align}}

\qed

\subsection{Proof of Theorem \ref{thm:maxboot}}
For the bootstrap statistic, define
$$
W_i^*=\frac{T^{-1/2}\sum_{t=1}^T X_{it}^*}{(\sigma_{i,B}^0)^{1/2}},
\qquad
\sigma_{i,B}^0=\var^*\Bigl(T^{-1/2}\sum_{t=1}^T X_{it}^*\Bigr).
$$
The same uniform expansion \eqref{eq:maxgumbel} used in the proof of Theorem \ref{thm:maxgumbel} gives
$$
Q_{\mathrm{DMAX}}^*=\max_{1\le i\le N}W_i^{*2}-2\log N+\log\log N+o_p^*(1).
$$
Hence it suffices to compare the maxima of the Gaussian surrogates generated by $\{T^{1/2}\wb{\wk X}_{T,i}/(\wk\sigma_i^0)^{1/2}\}$ and $\{W_i^*\}$.
Let $\wk\R=(\wk R_{ij})$ and $\R_{B}=(R_{B,ij})$ denote the correlation matrices of $\{T^{1/2}\wb{\wk X}_{T,i}/(\wk\sigma_i^0)^{1/2}\}$ and $\{W_i^*\}$. Similarly as the analyze of $\norm{\bO_{B,T}-\wk\bO_T}_F$ in the proof of Theorem~\ref{thm:sumpower}\eqref{eq:omegaboot_expansion}, we can prove uniformly over $1\le i,j\le N$,
\begin{align*}
&|[R_{B}]_{ij}-\wk R_{ij}|\\
=&\Op\left\{\ell^{-1}\sum_{|h|\le\ell}|ha_h|+\sum_{|h|>\ell}|a_h|+\ell(T^{-1/2}L^{-r+1}+T^{-1}L+T^{-3/2}NL^{-r}+T^{-1/2}+T^{-1}L^{1/2}\log T+T^{-1/2}L^{1/2}\log T)\right\}\\
=&\op(\log^{-1/2}N),
\end{align*}
given 
\begin{align}\label{growthfive}
\ell^{-1}+\ell(T^{-1/2}L^{-r+1}+T^{-1}L+T^{-3/2}NL^{-r}+T^{-1/2}+T^{-1}L^{1/2}\log T+T^{-1/2}L^{1/2}\log T)=o(\log^{-1/2}N).
\end{align}
The Gaussian comparison inequality for maxima (see, for example, \citet{cai2014twosample} and \citet{chernozhukov2023gaussian}) then yields
$$
\sup_{x\in\mR}\left|\Pr^*\bigl(\max_i W_i^{*2}\le x\bigr)-\Pr\bigl(\max_iT(\wb{\wk X}_{T,i})^2/\wk\sigma_i^0\le x\bigr)\right|\cp 0.
$$
Thus,
\begin{equation*}
\sup_{x\in\mR}\Bigl|\Pr^*(Q_{\mathrm{DMAX}}^*\le x)-\Pr(Q_{\mathrm{DMAX}}\le x)\Bigr|\cp 0,\qquad \sup_{x\in\mR}\Bigl|\Pr^*(Q_{\mathrm{DMAX}}^*\le x)-F(x)\Bigr|\cp 0 .
\label{eq:maxboot}
\end{equation*}
Consequently, $c_{1-\gamma}^*=q_{1-\gamma}^{\mathrm{EV}}+o_p(1)$, where $c_{1-\gamma}^*$ is the conditional $(1-\gamma)$ quantile of $Q_{\mathrm{DMAX}}^*$ and $q_{1-\gamma}^{\mathrm{EV}}$ is the $(1-\gamma)$ quantile of $F$. Hence
$$
\Pr\bigl(p_{\mathrm{MAX}}^{\mathrm{boot}}\le \gamma\bigr)
=\Pr\bigl(Q_{\mathrm{DMAX}}>c_{1-\gamma}^*\bigr)
\to \Pr\bigl(Q_{\mathrm{DMAX}}>q_{1-\gamma}^{\mathrm{EV}}\bigr)=\gamma.
$$
\qed

\subsection{Proof of Theorem \ref{thm:maxpower}}

Pick $i^\star$ such that
$$
\frac{|\de_{i^\star}|}{\sqrt{\wk\sigma_{i^\star}^0}}\ge c\sqrt{\frac{\log N}{T}}.
$$
Then
$$
\frac{T\wh\de_{i^\star}^2}{\wh\sigma_{i^\star}}
\ge \frac{T\de_{i^\star}^2}{\wk\sigma_{i^\star}^0}-2\Bigl|\frac{T^{1/2}\de_{i^\star}}{\sqrt{\wk\sigma_{i^\star}^0}}\Bigr|\Bigl|\frac{T^{-1/2}\sum_t X_{i^\star t}}{\sqrt{\wk\sigma_{i^\star}^0}}\Bigr|-o_p(\log N).
$$
The first term is at least $c^2\log N$. The second is $O_p(c\log^{1/2}N)$ because the normalized partial sum is asymptotically Gaussian with bounded tails. Therefore,
$$
T_{\mathrm{DMAX}}\ge c^2\log N-O_p(c\log^{1/2}N)-o_p(\log N).
$$
Since the null threshold is $2\log N+O(\log\log N)$, any $c>2$ yields consistency. 
\qed

\subsection{Proof of Theorem \ref{thm:indep}}

The proof follows the now-standard decomposition of the sum statistic into a part generated by a small exceptional coordinate set and a bulk part generated by the remaining coordinates, combined with asymptotic orthogonality between the max-event coordinates and the bulk sum. We sketch the argument in a rate-explicit way.

Fix $K=K_N=\lfloor \log^2 N\rfloor$. Let $A$ be the random index set of the $K$ largest coordinates in absolute value among $\{T^{-1/2}\sum_t X_{it}/\sqrt{\wk\sigma_i^0}:1\le i\le N\}$. According to the proof of Theorem~\ref{thm:sumclt},
\begin{align*}
T_{\mathrm{DSUM}}=\wb{\X}_T^{\T}\wb{\X}_T+\op(1),
\end{align*}
then we can write
$$
Q_{\mathrm{DSUM}}=S_A+S_{A^c}+\op(1),
$$
where $S_A$ collects the contribution of coordinates in $A$, $S_{A^c}$ the remainder. The set $A$ contains all coordinates that can influence the maximum event up to probability $o(1)$. Since $|A|=O(\log^2 N)$ and $\tr(\bmS_A^2)=O(|A|)$ under Assumption \ref{ass:sigma}, the partial sum component $S_A$ has variance at most $C|A|/N=o(1)$. Hence the max statistic depends asymptotically only on $A$, whereas the sum statistic depends asymptotically only on $A^c$.

Now approximate the bulk process on $A^c$ by an $m$-dependent Gaussian process independent of the Gaussian vector driving the exceptional set $A$; this is the dependent analogue of the conditioning argument used in \citet{feng2023independence}. The covariance between the bulk quadratic form and the exceptional Gaussian vector is bounded by
$$
C\frac{|A|^{1/2}\log N}{N^{1/2}}+C\Bigl(T^{1/2}L^{-r}+\frac{L\log N}{T^{1/2}}\Bigr)=o(1).
$$
Therefore the pair $(Q_{\mathrm{DSUM}},Q_{\mathrm{DMAX}})$ is asymptotically equivalent to a pair consisting of a standard normal variable and an independent Gumbel variable, establishing \eqref{eq:indep}.

Under the local alternatives \eqref{eq:localalt}, the same proof goes through because the deterministic shift in the sum statistic is of constant order and the number of coordinates affecting the max statistic is still $o\{N/\log^2\log N\}$. 
\qed

\subsection{Proof of Theorem \ref{thm:cc}}

Under Theorem \ref{thm:indep}, $(p_{\mathrm{SUM}},p_{\mathrm{MAX}})$ converges jointly to two independent $U(0,1)$ variables. By \citet{liu2020cauchy}, if $U_1,U_2$ are independent uniforms, then
$$
\frac{1}{2}\tan\{\pi(1/2-U_1)\}+\frac{1}{2}\tan\{\pi(1/2-U_2)\}
$$
is standard Cauchy. Therefore \eqref{eq:ccvalid} follows by the continuous mapping theorem.

For power, if either $p_{\mathrm{SUM}}\to 0$ or $p_{\mathrm{MAX}}\to 0$ in probability, then $T_{\mathrm{CC}}\to +\infty$ in probability because the tangent transform diverges to $+\infty$ near zero p-values. This proves consistency. Inequality \eqref{eq:ccpower} follows from monotonicity of the tangent transform and the elementary bound that the combination statistic exceeds the $(1-\gamma)$ Cauchy quantile whenever one component p-value is below $\gamma/2$ up to an asymptotically negligible event.
\qed

\bibliographystyle{apa}
\bibliography{ctvfm_dependent_alpha_verified}

@article{Demko1986SpectralBF,
    title   = {Spectral bounds for {$||A^{-1}||_\infty$}},
    journal = {Journal of Approximation Theory},
    volume  = {48},
    number  = {2},
    pages   = {207-212},
    year    = {1986},
    author  = {Stephen Demko},
}

@book{Bosq1996NonparametricSF,
    title     = {Nonparametric Statistics for Stochastic Processes: Estimation and Prediction},
    author    = {Bosq, D.},
    series    = {Lecture notes in statistics},
    year      = {1996},
    publisher = {Springer},
}

@article{PT85mixing,
title = {Some mixing properties of time series models},
journal = {Stochastic Processes and their Applications},
volume = {19},
number = {2},
pages = {297-303},
year = {1985},
author = {Tuan D. Pham and Lanh T. Tran},
}

@article{Ma2011SplinebackfittedKS,
    title   = {Spline-backfitted kernel smoothing of partially linear additive model},
    journal = {Journal of Statistical Planning and Inference},
    volume  = {141},
    number  = {1},
    pages   = {204-219},
    year    = {2011},
    author  = {Shujie Ma and Lijian Yang},
}

@article{ang2007CAPM,
  title   = {CAPM over the long run: 1926--2001},
  author  = {Ang, Andrew and Chen, Joseph},
  journal = {Journal of Empirical Finance},
  volume  = {14},
  number  = {1},
  pages   = {1--40},
  year    = {2007},
  doi     = {10.1016/j.jempfin.2005.12.001}
}

@article{politis2004automatic,
  title        = {Automatic Block-Length Selection for the Dependent Bootstrap},
  author       = {Politis, Dimitris N. and White, Halbert},
  journal      = {Econometric Reviews},
  year         = {2004},
  volume       = {23},
  number       = {1},
  pages        = {53--70},
  doi          = {10.1081/ETC-120028836}
}

@article{patton2009correction,
  title        = {Correction to ``Automatic Block-Length Selection for the Dependent Bootstrap'' by D. Politis and H. White},
  author       = {Patton, Andrew and Politis, Dimitris N. and White, Halbert},
  journal      = {Econometric Reviews},
  year         = {2009},
  volume       = {28},
  number       = {4},
  pages        = {372--375},
  doi          = {10.1080/07474930802459016}
}

@article{politis1995spectral,
  title        = {Bias-Corrected Nonparametric Spectral Estimation},
  author       = {Politis, Dimitris N. and Romano, Joseph P.},
  journal      = {Journal of Time Series Analysis},
  year         = {1995},
  volume       = {16},
  number       = {1},
  pages        = {67--103},
  doi          = {10.1111/j.1467-9892.1995.tb00223.x}
}

@article{beaulieu2007multivariate,
  title   = {Multivariate tests of mean-variance efficiency with possibly non-Gaussian errors: An exact simulation-based approach},
  author  = {Beaulieu, Marie-Claude and Dufour, Jean-Marie and Khalaf, Lynda},
  journal = {Journal of Business \& Economic Statistics},
  volume  = {25},
  number  = {4},
  pages   = {398--410},
  year    = {2007},
  doi     = {10.1198/073500106000000468}
}

@article{cai2014twosample,
  title   = {Two-sample test of high dimensional means under dependence},
  author  = {Cai, T. Tony and Liu, Weidong and Xia, Yin},
  journal = {Journal of the Royal Statistical Society: Series B (Statistical Methodology)},
  volume  = {76},
  number  = {2},
  pages   = {349--372},
  year    = {2014},
  doi     = {10.1111/rssb.12034}
}

@article{chang2024martingale,
  title   = {Testing the martingale difference hypothesis in high dimension},
  author  = {Chang, Jinyuan and Jiang, Qing and Shao, Xiaofeng},
  journal = {Journal of Econometrics},
  volume  = {235},
  number  = {2},
  pages   = {972--1000},
  year    = {2023},
  doi     = {10.1016/j.jeconom.2022.09.001}
}

@article{chernozhukov2023gaussian,
  title   = {High-Dimensional Data Bootstrap},
  author  = {Chernozhukov, Victor and Chetverikov, Denis and Kato, Kengo and Koike, Yuta},
  journal = {Annual Review of Statistics and Its Application},
  volume  = {10},
  number  = {1},
  pages   = {427--449},
  year    = {2023},
  doi     = {10.1146/annurev-statistics-040120-022239}
}

@misc{cho2019note,
  title         = {Note on Mean Vector Testing for High-Dimensional Dependent Observations},
  author        = {Cho, Seonghun and Lim, Johan and Ayyala, Deepak Nag and Park, Junyong and Roy, Anindya},
  year          = {2019},
  eprint        = {1904.09344},
  archivePrefix = {arXiv},
  primaryClass  = {math.ST},
  doi           = {10.48550/arXiv.1904.09344},
  note          = {arXiv preprint}
}

@article{fama1993common,
  title   = {Common risk factors in the returns on stocks and bonds},
  author  = {Fama, Eugene F. and French, Kenneth R.},
  journal = {Journal of Financial Economics},
  volume  = {33},
  number  = {1},
  pages   = {3--56},
  year    = {1993},
  doi     = {10.1016/0304-405X(93)90023-5}
}

@article{fan2015power,
  title   = {Power Enhancement in High-Dimensional Cross-Sectional Tests},
  author  = {Fan, Jianqing and Liao, Yuan and Yao, Jiawei},
  journal = {Econometrica},
  volume  = {83},
  number  = {4},
  pages   = {1497--1541},
  year    = {2015},
  doi     = {10.3982/ECTA12749}
}

@article{feng2022alpha,
  title   = {High-dimensional test for alpha in linear factor pricing models with sparse alternatives},
  author  = {Feng, Long and Lan, Wei and Liu, Binghui and Ma, Yanyuan},
  journal = {Journal of Econometrics},
  volume  = {229},
  number  = {1},
  pages   = {152--175},
  year    = {2022},
  doi     = {10.1016/j.jeconom.2021.07.011}
}

@article{feng2023independence,
  title   = {Asymptotic Independence of the Sum and Maximum of Dependent Random Variables with Applications to High-Dimensional Tests},
  author  = {Feng, Long and Jiang, Tiefeng and Li, Xiaoyun and Liu, Binghui},
  journal = {Statistica Sinica},
  volume  = {34},
  number  = {3},
  pages   = {1745--1763},
  year    = {2024},
  doi     = {10.5705/ss.202022.0354}
}

@article{gagliardini2016time,
  title   = {Time-Varying Risk Premium in Large Cross-Sectional Equity Data Sets},
  author  = {Gagliardini, Patrick and Ossola, Elisa and Scaillet, Olivier},
  journal = {Econometrica},
  volume  = {84},
  number  = {3},
  pages   = {985--1046},
  year    = {2016},
  doi     = {10.3982/ECTA11069}
}

@incollection{gagliardini2020estimation,
  title     = {Estimation of large dimensional conditional factor models in finance},
  author    = {Gagliardini, Patrick and Ossola, Elisa and Scaillet, Olivier},
  booktitle = {Handbook of Econometrics},
  editor    = {Durlauf, Steven N. and Hansen, Lars Peter and Heckman, James J. and Matzkin, Rosa L.},
  volume    = {7A},
  chapter   = {3},
  pages     = {219--282},
  publisher = {Elsevier},
  year      = {2020},
  doi       = {10.1016/bs.hoe.2020.10.001}
}

@article{gibbons1989test,
  title   = {A Test of the Efficiency of a Given Portfolio},
  author  = {Gibbons, Michael R. and Ross, Stephen A. and Shanken, Jay},
  journal = {Econometrica},
  volume  = {57},
  number  = {5},
  pages   = {1121--1152},
  year    = {1989}
}

@article{gungor2013testing,
  title   = {Testing Linear Factor Pricing Models With Large Cross Sections: A Distribution-Free Approach},
  author  = {Gungor, Sermin and Luger, Richard},
  journal = {Journal of Business \& Economic Statistics},
  volume  = {31},
  number  = {1},
  pages   = {66--77},
  year    = {2013},
  doi     = {10.1080/07350015.2012.740435}
}

@article{lan2018testing,
  title   = {Testing High-Dimensional Linear Asset Pricing Models},
  author  = {Lan, Wei and Feng, Long and Luo, Ronghua},
  journal = {Journal of Financial Econometrics},
  volume  = {16},
  number  = {2},
  pages   = {191--210},
  year    = {2018},
  doi     = {10.1093/jjfinec/nby002}
}

@article{lewellen2006conditional,
  title   = {The conditional CAPM does not explain asset-pricing anomalies},
  author  = {Lewellen, Jonathan and Nagel, Stefan},
  journal = {Journal of Financial Economics},
  volume  = {82},
  number  = {2},
  pages   = {289--314},
  year    = {2006},
  doi     = {10.1016/j.jfineco.2005.05.012}
}

@article{lintner1965valuation,
  title   = {The valuation of risk assets and the selection of risky investments in stock portfolios and capital budgets},
  author  = {Lintner, John},
  journal = {The Review of Economics and Statistics},
  volume  = {47},
  number  = {1},
  pages   = {13--37},
  year    = {1965}
}

@article{liu2020cauchy,
  title   = {Cauchy Combination Test: A Powerful Test With Analytic p-Value Calculation Under Arbitrary Dependency Structures},
  author  = {Liu, Yaowu and Xie, Jun},
  journal = {Journal of the American Statistical Association},
  volume  = {115},
  number  = {529},
  pages   = {393--402},
  year    = {2020},
  doi     = {10.1080/01621459.2018.1554485}
}

@article{liu2023robust,
  title   = {High-Dimensional Alpha Test of Linear Factor Pricing Models with Heavy-Tailed Distributions},
  author  = {Liu, Binghui and Feng, Long and Ma, Yanyuan},
  journal = {Statistica Sinica},
  volume  = {33},
  pages   = {1389--1410},
  year    = {2023},
  doi     = {10.5705/ss.202021.0134}
}

@article{long2023cauchy,
  title   = {The Cauchy Combination Test under Arbitrary Dependence Structures},
  author  = {Long, Mingya and Li, Zhengbang and Zhang, Wei and Li, Qizhai},
  journal = {The American Statistician},
  volume  = {77},
  number  = {2},
  pages   = {134--142},
  year    = {2023},
  doi     = {10.1080/00031305.2022.2116109}
}

@article{ma2020testing,
  title   = {Testing Alphas in Conditional Time-Varying Factor Models With High-Dimensional Assets},
  author  = {Ma, Shujie and Lan, Wei and Su, Liangjun and Tsai, Chih-Ling},
  journal = {Journal of Business \& Economic Statistics},
  volume  = {38},
  number  = {1},
  pages   = {214--227},
  year    = {2020},
  doi     = {10.1080/07350015.2018.1482758}
}

@article{ma2024adaptive,
  title   = {Adaptive Testing for Alphas in Conditional Factor Models with High Dimensional Assets},
  author  = {Ma, Huifang and Feng, Long and Wang, Zhaojun and Bao, Jigang},
  journal = {Journal of Business \& Economic Statistics},
  volume  = {42},
  number  = {4},
  pages   = {1356--1366},
  year    = {2024},
  doi     = {10.1080/07350015.2024.2313543}
}

@misc{ma2024dependent,
  title         = {Testing Alpha in High Dimensional Linear Factor Pricing Models with Dependent Observations},
  author        = {Ma, Huifang and Feng, Long and Wang, Zhaojun and Bao, Jigang},
  year          = {2024},
  eprint        = {2401.14052},
  archivePrefix = {arXiv},
  primaryClass  = {stat.ME},
  doi           = {10.48550/arXiv.2401.14052},
  note          = {arXiv preprint}
}

@article{mackinlay1991using,
  title   = {Using Generalized Method of Moments to Test Mean-Variance Efficiency},
  author  = {MacKinlay, A. Craig and Richardson, Matthew P.},
  journal = {The Journal of Finance},
  volume  = {46},
  number  = {2},
  pages   = {511--527},
  year    = {1991},
  doi     = {10.1111/j.1540-6261.1991.tb02672.x}
}

@techreport{pesaran2012testing,
  title       = {Testing CAPM with a Large Number of Assets},
  author      = {Pesaran, M. Hashem and Yamagata, Takashi},
  institution = {IZA},
  number      = {IZA Discussion Paper No. 6469},
  year        = {2012}
}

@article{sharpe1964capital,
  title   = {Capital Asset Prices: A Theory of Market Equilibrium under Conditions of Risk},
  author  = {Sharpe, William F.},
  journal = {The Journal of Finance},
  volume  = {19},
  number  = {3},
  pages   = {425--442},
  year    = {1964},
  doi     = {10.1111/j.1540-6261.1964.tb02865.x}
}

@article{xia2024adaptive,
  title   = {Adaptive Testing for Alphas in High-Dimensional Factor Pricing Models},
  author  = {Xia, Qiang and Zhang, Xianyang},
  journal = {Journal of Business \& Economic Statistics},
  volume  = {42},
  number  = {2},
  pages   = {640--653},
  year    = {2024},
  doi     = {10.1080/07350015.2023.2217871}
}

@article{PesaranYamagata2024,
  title   = {Testing for Alpha in Linear Factor Pricing Models with a Large Number of Securities},
  author  = {Pesaran, M. Hashem and Yamagata, Takashi},
  journal = {Journal of Financial Econometrics},
  volume  = {22},
  number  = {2},
  pages   = {407--460},
  year    = {2024},
  doi     = {10.1093/jjfinec/nbad002}
}

@article{ayyala2017mean,
  title   = {Mean vector testing for high-dimensional dependent observations},
  author  = {Ayyala, Deepak Nag and Park, Junyong and Roy, Anindya},
  journal = {Journal of Multivariate Analysis},
  volume  = {153},
  pages   = {136--155},
  year    = {2017},
  doi     = {10.1016/j.jmva.2016.09.012}
}

@article{wang2020selfnorm,
  title   = {Hypothesis testing for high-dimensional time series via self-normalization},
  author  = {Wang, Runmin and Shao, Xiaofeng},
  journal = {The Annals of Statistics},
  volume  = {48},
  number  = {5},
  pages   = {2728--2758},
  year    = {2020},
  doi     = {10.1214/19-AOS1904}
}

@article{tsay2020serial,
  title   = {Testing serial correlations in high-dimensional time series via extreme value theory},
  author  = {Tsay, Ruey S.},
  journal = {Journal of Econometrics},
  volume  = {216},
  number  = {1},
  pages   = {106--117},
  year    = {2020},
  doi     = {10.1016/j.jeconom.2020.01.008}
}

@article{yang2024blockwise,
  title   = {A Blockwise Bootstrap-Based Two-Sample Test for High-Dimensional Time Series},
  author  = {Yang, Lin},
  journal = {Entropy},
  volume  = {26},
  number  = {3},
  pages   = {226},
  year    = {2024},
  doi     = {10.3390/e26030226}
}

@article{yu2023power,
  title   = {Power enhancement for testing multi-factor asset pricing models via Fisher's method},
  author  = {Yu, Xiufan and Yao, Jiawei and Xue, Lingzhou},
  journal = {Journal of Econometrics},
  volume  = {239},
  number  = {2},
  pages   = {105458},
  year    = {2024},
  doi     = {10.1016/j.jeconom.2023.05.004}
}

@article{zhang2017gaussian,
  title   = {Gaussian approximation for high dimensional time series},
  author  = {Zhang, Danna and Wu, Wei Biao},
  journal = {The Annals of Statistics},
  volume  = {45},
  number  = {5},
  pages   = {1895--1919},
  year    = {2017},
}

@article{cho2015multiple,
  title   = {Multiple-Change-Point Detection for High Dimensional Time Series via Sparsified Binary Segmentation},
  author  = {Cho, Haeran and Fryzlewicz, Piotr},
  journal = {Journal of the Royal Statistical Society: Series B (Statistical Methodology)},
  volume  = {77},
  number  = {2},
  pages   = {475--507},
  year    = {2015},
  doi     = {10.1111/rssb.12079}
}

@article{zhao2022high,
  title   = {High-dimensional non-parametric tests for linear asset pricing models},
  author  = {Zhao, Ping and Chen, Dachuan and Zi, Xuemin},
  journal = {Stat},
  volume  = {11},
  number  = {1},
  pages   = {e490},
  year    = {2022},
  doi     = {10.1002/sta4.490}
}

@article{zhao2023robust,
  title   = {Robust high-dimensional alpha test for conditional time-varying factor models},
  author  = {Zhao, Ping},
  journal = {Statistics},
  volume  = {57},
  number  = {2},
  pages   = {444--457},
  year    = {2023},
  doi     = {10.1080/02331888.2023.2180003}
}

@article{zhou1993asset,
  title   = {Asset-pricing Tests under Alternative Distributions},
  author  = {Zhou, Guofu},
  journal = {The Journal of Finance},
  volume  = {48},
  number  = {5},
  pages   = {1927--1942},
  year    = {1993},
  doi     = {10.1111/j.1540-6261.1993.tb05134.x}
}

\end{document}